%% file: paper.tex
\documentclass{tlp}

\usepackage{amsmath,amssymb}
\usepackage{times,helvet,courier}
\usepackage{calc}
\usepackage{listings}

\newlength\listingnumberwidth
\input{macro}

\newtheorem{proposition}{Proposition}
\newtheorem{theorem}{Theorem}
\newtheorem{lemma}{Lemma}
\newcommand{\SPACY}{\par}
\newcommand{\EXT}[2]{#2} 

\EXT{\addtolength{\textfloatsep}{-2ex}}{\addtolength{\textfloatsep}{-3.3ex}}

\title{Complex Optimization in Answer Set Programming\EXT{\\--- Extended Version ---}{}}

\author[Martin Gebser
  \and
  Roland Kaminski
  \and
  Torsten Schaub]{%
  Martin Gebser
  and
  Roland Kaminski
  and
  Torsten Schaub\thanks{Affiliated with Simon Fraser University, Canada, and Griffith University, Australia.}%
  \\
  Institut f\"ur Informatik, Universit\"at Potsdam}

\EXT{\submitted{[n/a]}}{\submitted{[TBA]}}
\EXT{\revised{[n/a]}}{\revised{[TBA]}}
\EXT{\accepted{[n/a]}}{\accepted{[TBA]}}

\begin{document}

\maketitle

\input{abstract}

\input{introduction}

\input{background}

\input{meta}
\input{approach} 
\input{application}

\input{discussion}
\input{acknowledgments}
\EXT{\newpage}{}
\bibliographystyle{acmtrans} 

\input{bbl}
\end{document}

%% file: macro.tex
\newcommand{\naf}[1]{\ensuremath{{\sim}{#1}}} 

\newcommand{\minimize}[0]{\ensuremath{\mathtt{\#minimize}}}
\newcommand{\summ}[0]{\ensuremath{\mathtt{\#sum}}}

\newcommand{\head}[1]{\ensuremath{\mathit{head}(#1)}} 
\newcommand{\atom}[1]{\ensuremath{\mathit{atom}(#1)}} 
\newcommand{\body}[1]{\ensuremath{\mathit{body}(#1)}} 
\newcommand{\pbody}[1]{\poslits{\body{#1}}} 
\newcommand{\poslits}[1]{\ensuremath{#1^+}}

\newcommand{\reduct}[2]{\ensuremath{#1^{#2}}}
\newcommand{\tp}{\ensuremath{\mathcal{T}}}
\newcommand{\tps}[2]{\ensuremath{\tp_{\,#1}(#2)}}
\newcommand{\tpi}[3]{\ensuremath{\tp_{\,#1}^{\,#3}(#2)}}

\newcommand{\claspD}{\textit{claspD}}

\newcommand{\clasp}{\textit{clasp}}

\newcommand{\gringo}{\textit{gringo}}

\newcommand{\psmodels}{\textit{psmodels}}

\newcommand{\dlv}{\textit{dlv}}
\newcommand{\lparse}{\textit{lparse}}
\newcommand{\smodels}{\textit{smodels}}


%% file: abstract.tex
\begin{abstract}
  Preference handling and optimization are indispensable means 
  for addressing non-trivial applications in Answer Set Programming (ASP).
  However, their implementation becomes difficult whenever they bring about
  a significant increase in computational complexity.
  As a consequence,
  existing ASP systems do not offer 
  complex optimization capacities,
  supporting, for instance, inclusion-based minimization or Pareto efficiency.
  Rather, such complex criteria are typically addressed by resorting to dedicated
  modeling techniques, like \textit{saturation}.
  Unlike the ease of common ASP modeling, however,
  these techniques are rather involved and hardly usable by ASP laymen.
  We address this problem by developing a general implementation technique
  by means of meta-programming, thus reusing existing ASP systems to
  capture various forms of qualitative preferences among answer sets.
  In this way, complex preferences and optimization capacities become readily
  available for ASP applications. 
\end{abstract}


%% file: introduction.tex
\section{Introduction}\label{sec:introduction}

Preferences are often an indispensable means in modeling since they allow for
identifying preferred solutions among all feasible ones.
Accordingly, many forms of preferences have already found their way into
systems for Answer Set Programming (ASP; \cite{baral02a}).
For instance, \smodels\ provides optimization statements for expressing cost
functions on sets of weighted literals \cite{siniso02a}, and
\dlv\ \cite{dlv03a} offers weak constraints for the same purpose.
Further 
approaches
\cite{descto02a,eifalepf03a} allow for
expressing various types of preferences among rules.
Unlike this, no readily applicable implementation techniques are available for
qualitative preferences among answer sets, like inclusion minimality,
Pareto-based preferences as used in \cite{sakino00a,brnisy04a}, or
more complex combinations as proposed in \cite{brewka04b}.
This shortcoming is due to their higher expressiveness leading to a significant
increase in computational complexity, lifting decision problems (for normal logic programs)
from the first to the second level of the polynomial time hierarchy (cf.\ \cite{garjoh79}).
Roughly speaking, preferences among answer sets combine an $\mathit{NP}$ with a
$\mathit{coNP}$ problem.
The first one defines feasible solutions, while the second one ensures that
there are no better solutions according to the preferences at hand.
For implementing such problems, Eiter and Gottlob invented in~\citeyear{eitgot95a}
the \textit{saturation} technique, using the elevated complexity of
disjunctive logic programming.
In stark contrast to the ease of common ASP modeling
(e.g., strategic companies can be ``naturally'' encoded
 \cite{dlv03a} in disjunctive ASP), however,
the saturation technique is rather involved and hardly usable by ASP laymen. 

For taking this burden of intricate modeling off the user,
we propose a general, saturation-based implementation
technique capturing various forms of qualitative preferences among answer sets.
This is driven by the desire to guarantee immediate availability and thus to
stay within the realm of ASP rather than to build separate (imperative) components.
To this end, we take advantage of recent advances in ASP grounding technology,
admitting an easy use of meta-modeling techniques.
The idea is to reinterpret existing optimization statements in order
to express complex preferences among answer sets.
While, for instance in \smodels, the meaning of $\minimize$ is to
compute answer sets incurring minimum costs,
we may alternatively use it for selecting inclusion-minimal ones.
In contrast to the identification of minimal models,
investigated by Janhunen
and Oikarinen in~\citeyear{janoik04a,oikjan08b},
a major challenge lies in guaranteeing the stability property of implicit
counterexamples, which must be more preferred answer sets rather than (arbitrary) models.
For this purpose, we develop a refined meta-program qualifying answer sets as
viable counterexamples.
Unlike the approach of Eiter and Polleres \citeyear{eitpol06a},
our encoding avoids ``guessing'' a level mapping
to describe the formation of a counterexample,
but directly denies models for which there is no such construction.
Notably, our meta-programs apply to (reified) \emph{extended} logic programs \cite{siniso02a},
possibly including choice rules and $\summ$ constraints,
and we are unaware of any existing meta-encoding of their answer sets,
neither as candidates nor as counterexamples refuting optimality.


%% file: background.tex
\section{Background}\label{sec:background}

We consider extended logic programs \cite{siniso02a}
allowing for (proper) disjunctions in heads of rules \cite{gellif91a}.
A \emph{rule}~$r$ is of the following form:
\begin{align*} 
H \leftarrow B_1,\dots,B_m,\naf B_{m+1},\dots,\naf B_n.
\end{align*}
By $\head{r}=H$ and $\body{r}=\{B_1,\dots,B_m,\naf B_{m+1},\dots,\naf B_n\}$,
we denote the \emph{head} and the \emph{body} of~$r$, respectively,
where ``$\naf{}$'' stands for default negation.
The head~$H$ is a disjunction $a_1\vee\dots\vee a_k$ over
\emph{atoms} $a_1,\dots,a_k$, belonging to some alphabet~$\mathcal{A}$,
or a $\summ$ constraint 
$L\, \summ[\ell_1=w_1,\dots,\ell_k=w_k]\, U$.
In the latter, 
$\ell_i=a_i$ or $\ell_i=\naf a_i$ 
is a \emph{literal} and $w_i$ a \emph{non-negative} integer \emph{weight} for $a_i\in\mathcal{A}$ and $1\leq i\leq k$;
$L$ and $U$ are integers providing a lower and an upper bound.
Either or both of $L$ and $U$ can be omitted,
in which case they are identified with the (trivial) bounds $0$ and~$\infty$, respectively.
A rule~$r$ such that $\head{r}=\bot$ ($H$ is the empty disjunction) is
an \emph{integrity constraint}.
Each body component $B_i$ is either an atom or a $\summ$ constraint for $1\leq i\leq n$.
If $\body{r}=\emptyset$,
$r$ is called a \emph{fact},
and we skip ``$\leftarrow$'' when writing facts below.
For a set $\{B_1,\dots,B_m,\naf B_{m+1},\dots,\naf B_n\}$,
a disjunction $a_1\vee\dots\vee a_k$, and a 
$\summ$ constraint $L\, \summ[\ell_1=w_1,\dots,\ell_k=w_k]\, U$,
we let $\poslits{\{B_1,\dots,B_m,\naf B_{m+1},\dots,\naf B_n\}}=\{B_1,\dots,B_m\}$,
$\poslits{(a_1\vee\dots\vee a_k)}=\{a_1,\dots,a_k\}$, and
$\poslits{(L\, \summ[\ell_1=w_1,\dots,\ell_k=w_k]\, U)}=
 [\ell_i=w_i \mid 1\leq i\leq k,\ell_i\in\mathcal{A}]$.
Note that the elements of a $\summ$ constraint
form a multiset, possibly containing duplicates.
For some $S=\{a_1,\dots,a_k\}$ or 
$S=[a_1=w_1,\dots,a_k=w_k]$, we define $\atom{S}=\{a_1,\dots,a_k\}$.

A (Herbrand) \emph{interpretation} is represented by the set~$X\subseteq\mathcal{A}$
of its entailed atoms.
The satisfaction relation ``$\models$'' on rules~$r$ is inductively defined as follows:
\begin{itemize}
\item $X\models\naf B$ if $X\not\models B$,
\item $X\models (a_1\vee\dots\vee a_k)$ if
      $\{a_1,\dots,a_k\}\cap X\neq\emptyset$,
\item $X\models (L\, \summ[\ell_1=w_1,\dots,\ell_k=w_k]\, U)$ if
      $L\leq \sum_{1\leq i\leq k,X\models\ell_i}w_i \leq U$,
\item $X\models \body{r}$ if
      $X\models \ell$ for all $\ell\in\body{r}$, and
\item $X\models r$      if
      $X\models \head{r}$ or
      $X\not\models \body{r}$.
\end{itemize}
A \emph{logic program}~$\Pi$ is a set of rules~$r$,
and $X$ is a \emph{model} of $\Pi$ if $X\models r$ for every $r\in\Pi$.
The reduct of the head~$H$ of a rule~$r$ wrt~$X$ is
$\reduct{H}{X}=\{a_1\vee\dots\vee a_k\}$ if $H=a_1\vee\dots\vee a_k$, and
$\reduct{H}{X}=\atom{\poslits{H}}\cap X$ if $H=L\, \summ[\ell_1=w_1,\dots,\ell_k=w_k]\, U$.
Furthermore, the reduct of some (positive) body element $B\in\poslits{\body{r}}$ is
$\reduct{B}{X}=B$ if $B\in\mathcal{A}$,
and
$\reduct{B}{X}=
 \big(L-\sum_{1\leq i\leq k,\ell_i=\naf a_i,a_i\notin X}w_i\big)\, \summ\:\poslits{B}$
if $B=L\, \summ[\ell_1=w_1,\dots,\ell_k=w_k]\, U$.
The \emph{reduct} of~$\Pi$ wrt~$X$ is the following logic program:
\pagebreak[1]%
\begin{align*}
\reduct{\Pi}{X}=
\big\{
&
H \leftarrow \reduct{B_1}{X},\dots,\reduct{B_m}{X}
\mid
{} \\ & \qquad
r\in\Pi,X\models\body{r},H\in\reduct{\head{r}}{X},\poslits{\body{r}}=\{B_1,\dots,B_m\}
\big\}.
\end{align*}
That is, for all rules $r\in\Pi$ whose bodies are satisfied wrt~$X$,
the reduct is obtained by replacing $\summ$ constraints in heads
with individual atoms belonging to~$X$ and by
eliminating negative components in bodies,
where lower bounds of residual $\summ$ constraints (with trivial upper bounds)
are reduced accordingly.
Finally, 
$X$ is an \emph{answer set} of~$\Pi$ if $X$ is a model of~$\Pi$
such that no proper subset of~$X$ is a model of~$\reduct{\Pi}{X}$.
In view of the latter condition,
note that an answer set is a \emph{minimal} model of its own reduct.

The definition of answer sets provided above applies to
logic programs containing extended constructs ($\summ$ constraints)
under ``choice semantics'' \cite{siniso02a},
while additionally allowing for disjunctions under minimal-model semantics
(wrt a reduct).
We use these features to
embed extended constructs of an object program into a disjunctive
meta-program, so that their combination yields optimal answer sets
of the object program.
To this end,
we reinterpret $\minimize$ statements of the following form:
\begin{align}\label{eq:minimize}
\minimize[\ell_1=w_1@J_1,\dots,\ell_k=w_k@J_k].  
\end{align}
Like with $\summ$ constraints, every $\ell_i$ is a literal and every $w_i$
an integer weight for $1\leq i\leq k$,
while $J_i$ additionally provides an integer \emph{priority level}.\footnote{%
Explicit priority levels are supported in recent versions of the grounder
\gringo\ \cite{potasscoManual}. 
This avoids a dependency of priorities on input order,
which is considered by \lparse\ \cite{lparseManual}
if several $\minimize$ statements are provided.
Priority levels are also supported by \dlv\ \cite{dlv03a}
in weak constraints.
Furthermore, we admit negative weights in $\minimize$ statements,
where they cannot raise semantic problems (cf.\ \cite{ferraris05a}) going along
with the rewriting of $\summ$ constraints suggested in \cite{siniso02a}.}
Priorities allow for representing a sequence
of lexicographically ordered $\minimize$ objectives,
where greater levels are more significant than smaller ones.
By default,
a $\minimize$ statement distinguishes optimal answer sets of a program~$\Pi$
in the following way.
For any~$X\subseteq\mathcal{A}$ and integer $J$,
let $\Sigma^X_J$ denote the sum of weights $w$ over all occurrences of
weighted literals $\ell=w@J$ in~(\ref{eq:minimize}) such that $X\models\ell$. 
An answer set~$X$ of~$\Pi$ is dominated if there is an answer set~$Y$ of~$\Pi$
such that $\Sigma^Y_J<\Sigma^X_J$ and $\Sigma^Y_{J'}=\Sigma^X_{J'}$ for all $J'>J$,
and optimal otherwise.

In the following,
we assume that every logic program is accompanied with
one (possibly empty) $\minimize$ statement of the form~(\ref{eq:minimize}).
Instead of the default semantics,
we consider Pareto efficiency 
wrt priority levels~$J$, weights~$w$, and several distinct optimization criteria.
In view of this, we use levels for inducing a lexicographic order, while weights
are used for grouping literals (rather than summation).
Pareto improvement then builds upon a two-dimensional structure of
orderings among answer sets, induced by $J$ and $w$.
In turn, each such pairing is associated with some of the following orderings.
By $Y\leq^w_J X$, we denote that the cardinality of the multiset of
occurrences of $\ell=w@J$ in~(\ref{eq:minimize}) 
such that $Y\models\ell$
is not greater than the one of the
corresponding multiset for $X\models\ell$.
Furthermore, 
we write $Y\subseteq^w_J X$ if, for any weighted literal $\ell=w@J$
occurring in~(\ref{eq:minimize}),  
$Y\models\ell$ implies $X\models\ell$.
\EXT{Finally, we denote by $Y\preceq^w_J X$
that $Y$ is preferable to~$X$ \cite{sakino00a}}%
{As detailed in the extended version of this paper \cite{gekasc11A}, we additionally consider 
the approach of \cite{sakino00a} and  denote by $Y\preceq^w_J X$
that $Y$ is preferable to~$X$}
according to a (given) preference relation $\preceq$
among literals~$\ell$ such that $\ell=w@J$
occurs in~(\ref{eq:minimize}).  
Given a logic program~$\Pi$ and a collection~$M$ of relations of the form
$\diamond^{w}_{J}$ for priority levels~$J$, weights~$w$, and $\diamond\in\{\leq,\subseteq,\preceq\}$,
an answer set~$Y$ of~$\Pi$ dominates an answer set~$X$ of~$\Pi$ wrt~$M$
if there are a priority level~$J$ and a weight~$w$ such that
$X \diamond^w_J Y$ does not hold for $\diamond^w_J\in M$, while
$Y \diamond^{w'}_{J'} X$ holds for all $\diamond^{w'}_{J'}\in M$ where $J' \geq J$.
In turn, an answer set~$X$ of~$\Pi$ is \emph{optimal} wrt~$M$
if there is no answer set~$Y$ of~$\Pi$ that dominates~$X$ wrt~$M$.

As an example, consider the following program, referred to by $\Pi_0$:
\begin{align}
\label{ex:one}     1\ \{ p, t \}\phantom{\ 1} &\ \leftarrow\ 1\ \{ r, s, \naf{t} \}\ 2.\\
\label{ex:two}        \{ q, r \}         \ 1  &\ \leftarrow\ 1\ \{ p, t \}.\\
\label{ex:tri}  s\phantom{, r \}         \ 1} &\ \leftarrow\ \naf{q}, \naf{r}.
\end{align}
This program has five answer sets, viz.\ 
$\{p,q\}$, $\{p,r\}$, $\{p,s\}$, $\{p,s,t\}$, and $\{s,t\}$.
(Sets $\{a_1,\dots,a_k\}$ in~(\ref{ex:one}) and~(\ref{ex:two}) are used as shorthands
 for $\summ[a_1=1,\dots,a_k=1]$.)
%
In addition, let $\Pi_1$ denote the union of~$\Pi_0$ with
the following $\minimize$ statement:
\begin{align}
\label{ex:for}&  \mathtt{\#minimize}[ p = 1 @ 1, q = 1 @ 1, r = 1 @ 1, s = 1 @ 1 ].
\end{align}
This statement specifies that all atoms of $\Pi_0$ except for~$t$ 
are subject to minimization.
Passing~$\Pi_1$ 
to \gringo\ and an answer set solver like \smodels\ yields the single
$\leq^1_1$-minimal answer set $\{s,t\}$.
Note, however, that $\Pi_0$ has three $\subseteq^1_1$-minimal answer sets,
namely $\{p,q\}$, $\{p,r\}$, and $\{s,t\}$.
They cannot be computed directly from~$\Pi_1$
via any 
available ASP system.

We implement the complex optimization criteria
described above 
by meta-interpretation in disjunctive ASP.
For transparency, we provide meta-programs as true ASP code in the 
first-order input language of \gringo\ \cite{potasscoManual},
including
\texttt{not} and \texttt{|} as tokens for 
$\naf{}$ and $\vee$, respectively,
as well as \texttt{\{$a_1$,$\dots$,$a_k$\}} as shorthand for
\texttt{$\summ$[$a_1$=1,$\dots$,$a_k$=1]}.
Further constructs are informally introduced by need in the remainder of this paper.
Note that our (disjunctive) meta-programs apply to an extended object program
that does not include proper disjunctions (over more than one atom).
Unless stated otherwise,
we below 
use the term \emph{extended} program 
to refer to a logic program without proper disjunctions.


%% file: meta.tex
\section{Basic Meta-Modeling}\label{sec:meta}

For reinterpreting $\minimize$ statements by means of ASP,
we take advantage of recent advances in ASP grounding,
admitting an easy use of meta-modeling techniques.
To be precise, we rely upon the unrestricted usage of function symbols and
program reification as provided by \gringo\ \cite{potasscoManual}.
The latter allows for turning an input program 
along with a $\minimize$ statement into facts 
representing the structure of their ground instantiation\EXT{ via
fixed sets of predicates and function symbols}{}.

For illustrating the format output by \gringo, 
consider the facts in Line~1--15 of Listing~\ref{fig:reified},
obtained by calling \gringo\ with option \texttt{{-}-reify} on program $\Pi_0$.
\begin{lstlisting}[float=t,caption={Facts describing a reified extended logic program.},captionpos=b,label=fig:reified]
rule(pos(sum(1,0,2)),pos(conjunction(0))). % 1 { p, t } :- 1 { r, s, not t } 2.
wlist(0,0,pos(atom(p)),1). wlist(0,1,pos(atom(t)),1).
set(0,pos(sum(1,1,2))).
wlist(1,0,pos(atom(r)),1). wlist(1,1,pos(atom(s)),1). wlist(1,2,neg(atom(t)),1).

rule(pos(sum(0,2,1)),pos(conjunction(1))). % { q, r } 1 :- 1 { p, t }.
wlist(2,0,pos(atom(q)),1). wlist(2,1,pos(atom(r)),1).
set(1,pos(sum(1,0,2))).

rule(pos(atom(s)),pos(conjunction(2))).    % s :- not q, not r.
set(2,neg(atom(q))). set(2,neg(atom(r))).

scc(0,pos(atom(p))). scc(0,pos(atom(r))). scc(0,pos(atom(t))).
scc(0,pos(conjunction(0))). scc(0,pos(sum(1,1,2))).
scc(0,pos(conjunction(1))). scc(0,pos(sum(1,0,2))).

minimize(1,3). % #minimize [ p = 1 @ 1, q = 1 @ 1, r = 1 @ 1, s = 1 @ 1 ].
wlist(3,0,pos(atom(p)),1). wlist(3,1,pos(atom(q)),1).
wlist(3,2,pos(atom(r)),1). wlist(3,3,pos(atom(s)),1).
\end{lstlisting}
Let us detail the representation of the rule in~(\ref{ex:one}) inducing the facts in
Line~1--4.
The predicate \texttt{rule/2} is used to link the rule head and body.
By convention, both are positive rule elements, as indicated 
via the functor \texttt{pos/1}.
Furthermore,
the term \texttt{sum(1,0,2)} tells us that the head
is a $\summ$ constraint with lower bound~\texttt{1} and (trivial) upper bound~\texttt{2} over
a list labeled~\texttt{0} of weighted literals.
In fact, the included literals are provided via the facts over \texttt{wlist/4}
given in Line~2, whose first arguments are~\texttt{0}.
While the second arguments, \texttt{0} and~\texttt{1}, are simply indexes
(enabling the representation of duplicates in multisets),
the third ones provide literals, \texttt{p} and~\texttt{t},
each having the (default) weight~\texttt{1}, as given in the fourth arguments.
%
%
Again by convention, the body of each rule is a conjunction,
where the term \texttt{conjunction(0)} in Line~1 refers to the
set labeled~\texttt{0}. 
Its single element,
a positive $\summ$ constraint with lower bound~\texttt{1} and
upper bound~\texttt{2} over a list labeled~\texttt{1},
is provided by the fact in Line~3.
The corresponding weighted literals are described by the facts in Line~4;
observe that the negative literal~\texttt{not\;t} is represented
in terms of the functor \texttt{neg/1}, applied to \texttt{atom(t)}.
The rules in~(\ref{ex:two}) and~(\ref{ex:tri}) are represented analogously in Line~6--8
and 10--11, respectively.
It is still interesting to note that recurrences of lists of weighted literals (and sets)
can reuse labels introduced before, as done in Line~8 by referring to~\texttt{0}.
In fact, \gringo\ identifies repetitions of structural entities and reuses labels.
In addition to the rules of~$\Pi_0$, 
the elements of non-trivial strongly connected components of its positive dependency
graph (cf.~(\ref{eq:dependency}) below) are provided in Line~13--15.
Albeit their usage is explained in the next section,
note already that the members of the only such component, labeled~\texttt{0},
include atoms as well as (positive) body elements, i.e., conjunctions and
$\summ$ constraints, 
connecting the component.
Indeed, the existence of facts over \texttt{scc/2}
tells us that~$\Pi_0$ is not tight (cf.\ \cite{fages94a}).

Now, we may compute all five answer sets of~$\Pi_0$ (given in \texttt{p0.lp}) by
combining the facts in~Line~1--15 of Listing~\ref{fig:reified} 
with the basic meta-program in Listing~\ref{fig:meta} (\texttt{meta.lp}):
\footnote{Following Unix customs, 
  the minus symbol ``\texttt{-}'' stands for the output of ``\texttt{gringo {-}-reify p0.lp}.''}%
\begin{lstlisting}[float=t,caption={Basic meta-program (\texttt{\small meta.lp}) for reified extended logic programs.},captionpos=b,label=fig:meta]
% extract rule elements

litb(B) :- rule(_,B).
litb(E) :- litb(pos(conjunction(S))), set(S,E).
litb(E) :- eleb(sum(_,S,_)), wlist(S,_,E,_).

eleb(P) :- litb(pos(P)).
eleb(N) :- litb(neg(N)).

elem(E) :- eleb(E).
elem(E) :- rule(pos(E),_).
elem(P) :- rule(pos(sum(_,S,_)),_), wlist(S,_,pos(P),_).
elem(N) :- rule(pos(sum(_,S,_)),_), wlist(S,_,neg(N),_).

% generate answer set from reified rules

hold(conjunction(S)) :- eleb(conjunction(S)),
               hold(P)     : set(S,pos(P)),
           not hold(N)     : set(S,neg(N)).
hold(sum(L,S,U))     :- eleb(sum(L,S,U)),
  L #sum [     hold(P) = W : wlist(S,Q,pos(P),W),
           not hold(N) = W : wlist(S,Q,neg(N),W) ] U.

hold(atom(A))        :- rule(pos(atom(A)),   pos(B)), hold(B).
L #sum [     hold(P) = W : wlist(S,Q,pos(P),W),
         not hold(N) = W : wlist(S,Q,neg(N),W) ] U
                     :- rule(pos(sum(L,S,U)),pos(B)), hold(B).
                     :- rule(pos(false),     pos(B)), hold(B).

% project output to atoms of answer set

#hide. #show hold(atom(A)).
\end{lstlisting}
\begin{lstlisting}[numbers=none,basicstyle=\ttfamily,aboveskip=\smallskipamount,belowskip=\smallskipamount]
gringo --reify p0.lp | gringo meta.lp - | clasp 0
\end{lstlisting}
Each answer set of the meta-program applied to a reified program
corresponds to an answer set of the reified program.
More precisely, a set~$X$ of atoms is an answer set of the reified program
iff the meta-program yields an answer set~$Y$ such
that $X = \{a \mid \text{\texttt{hold(atom($a$))}}\in Y\}$, e.g.,
\texttt{hold(atom($q$))} stands for~$q$.
As indicated in the comments (preceded by \texttt{\%}),
our meta-program consists of three parts.
Among the rule elements extracted in Line~3--13,
only those occurring within bodies, identified via \texttt{eleb/1},
are relevant to the generation of answer sets specified in Line~17--28.
(Additional head elements, given by \texttt{elem/1}, are of interest in the next section.)
In fact,
answer set generation follows the structure of reified programs,
identifying conjunctions and
$\summ$ constraints that hold\footnote{%
 The ``\texttt{:}'' connective
 expands to the list of all instances of its left-hand side
 such that corresponding instances of literals on the right-hand side hold
 (cf.\ \cite{lparseManual} and \cite{potasscoManual}).}
to further derive atoms occurring in rule heads, 
either singular or within $\summ$ constraints (cf.\ Line~24--27).
Line~28 deals with integrity constraints
represented via the constant \texttt{false} in heads of reified rules.
The last part in Line~32 restricts the output of the meta-program's
answer sets to the representations of original input atoms.

Finally, note that \texttt{meta.lp} does not inspect facts
representing a reified $\minimize$ statement, such as the ones
in Line 17--19 of Listing~\ref{fig:reified} stemming from the statement in~(\ref{ex:for}).
Such facts over \texttt{minimize/2} provide a priority level
as the first argument and the label of a list of weighted literals,
like the ones referred to from within terms of functor \texttt{sum/3},
as the second argument.
Rather than simply mirroring the   standard meaning of $\minimize$ statements
(by encoding them analogously to rules; cf.\ Line~17--28 of Listing~\ref{fig:meta}),
we support flexible customizations.
In fact, the next section presents our meta-programs implementing
preference relations and Pareto efficiency, as described in the background.



%% file: approach.tex
\section{Advanced Meta-Modeling}
\label{sec:approach}

Given the reification of extended logic programs and the encoding
of their answer sets in \texttt{meta.lp},
our approach to complex optimization is based on the idea
that an answer set generated via \texttt{meta.lp} is optimal (and
thus acceptable) only if it is not dominated by any other answer set.
For implementing our approach, we exploit the capabilities of 
disjunctive ASP to compactly represent the space of all potential
counterexamples, viz.\ answer sets dominating a candidate answer set at hand.
To this end, we encode the subtasks of
\begin{enumerate}
\item guessing an answer set as a potential counterexample and
\item verifying that the counterexample dominates a candidate answer set.
\end{enumerate}
A candidate answer set passes both phases if it turns out to be
infeasible to guess a counterexample that dominates it.
For expressing the non-existence of counterexamples,
we make use of an error-indicating atom \texttt{bot} and saturation \cite{eitgot95a},
deriving all atoms representing the space of counterexamples from \texttt{bot}.
Since the semantics of disjunctive ASP is based on minimization,
saturation makes sure that \texttt{bot} is derived only if it is inevitable, i.e.,
if it is impossible to construct a counterexample.
However,
via an integrity constraint, we can stipulate \texttt{bot} 
(and thus the non-existence of counterexamples) to hold,
yet without providing any derivation of \texttt{bot}.
In view of such a constraint and saturation, a successful candidate answer set
is accompanied by all atoms representing counterexamples.
Given that the reduct drops negative literals,
the necessity that all atoms representing counterexamples are true implies that
we cannot use their default negation in any meaningful way.
Hence, we below encode potential counterexamples, i.e., answer sets of extended programs,
and (non-)dominance of a candidate answer set in disjunctive ASP
without taking advantage of default negation (used in \texttt{meta.lp}).

For encoding the first subtask of guessing a counterexample, we
rely on a characterization of answer sets in terms of
an \emph{immediate consequence operator}~$\tp$ (cf.\ \cite{lloyd87}),
defined as follows for a logic program~$\Pi$ and a set $X\subseteq\mathcal{A}$
of atoms:
\( 
\tps{\Pi}{X}
= 
\{ 
  \head{r}
  \mid 
  r\in\Pi,
  X\models \body{r}
\}
\). 
Furthermore,
an iterative version of~$\tp$ can be defined in the following way:
\( 
\tpi{\Pi}{X}{0}
= 
X
\)
and
\(
\tpi{\Pi}{X}{i+1}
= 
\tpi{\Pi}{X}{i}\cup\tps{\Pi}{\tpi{\Pi}{X}{i}}
\). 
In the context of an extended program~$\Pi$,
possibly including choice rules, default negation, and upper bounds of weight constraints,
we are interested in the least fixpoint of~$\tp$ applied wrt
the reduct~$\reduct{\Pi}{X}$.
Since a fixpoint is reached in at most $|\atom{\Pi}|$ applications of~$\tp$,
where $\atom{\Pi}\subseteq\mathcal{A}$ denotes the set of atoms occurring in~$\Pi$,
the least fixpoint is
given by $\tpi{\reduct{\Pi}{X}}{\emptyset}{|\atom{\Pi}|}$.
As pointed out in \cite{liuyou10a},
a model~$X$ of an extended program~$\Pi$ is an answer set of~$\Pi$
iff $\tpi{\reduct{\Pi}{X}}{\emptyset}{|\atom{\Pi}|}=X$.
Furthermore,
Liu and You \citeyear{liuyou10a}
show that~$X$ violates the loop formula of some atom or loop
if $X$ is a model, but not an answer set of~$\Pi$.
This property motivates a ``localization'' of~$\tp$
on the basis of (circular) positive dependencies.

The \emph{(positive) dependency graph} of an extended program~$\Pi$ is
given by the following pair of nodes and directed edges:
%
\begin{align}\label{eq:dependency}
\left(
\atom{\Pi},
\{
  (a,b)\mid 
  r\in\Pi,a\in\atom{\poslits{\head{r}}},B\in\pbody{r},b\in\atom{\poslits{B}}
\}
\right)\!.\!\!
\end{align}
A strongly connected component (SCC) is a maximal subgraph of the 
dependency graph of~$\Pi$ such that all nodes are pairwisely connected via paths.
An SCC is trivial if it does not contain any edge, and non-trivial otherwise.
Note that the SCCs of the dependency graph of~$\Pi$ induce a partition
of $\atom{\Pi}$ such that every atom and every loop of~$\Pi$ is contained in
some part.
Hence, we can make use of the partition to apply~$\tp$ separately to each part.
\begin{proposition}\label{prop:tp}
Let $\Pi$ be an extended logic program,
$C_1,\dots,C_k$ be the sets of atoms belonging
to the SCCs of the dependency graph of~$\Pi$,
and $X\subseteq\atom{\Pi}$.
\SPACY{}
Then, we have that
$\tpi{\reduct{\Pi}{X}}{\emptyset}{|\atom{\Pi}|}=X$
iff
$\bigcup_{1\leq j\leq k}
(\tpi{\reduct{\Pi}{X}}{X\setminus C_j}{|C_j|}\cap C_j)=X$.
\end{proposition}
\EXT{%
\begin{proof}[Proof (Sketch)]
It is clear that
$\tpi{\reduct{\Pi}{X}}{\emptyset}{|\atom{\Pi}|}=X$
implies 
$\tpi{\reduct{\Pi}{X}}{X\setminus C_j}{|C_j|}\cap C_j=X\cap C_j$
for every $1\leq j\leq k$.
Hence,
we only need to show that
$\bigcup_{1\leq j\leq k}
(\tpi{\reduct{\Pi}{X}}{X\setminus C_j}{|C_j|}\cap C_j)=X$
implies 
$\tpi{\reduct{\Pi}{X}}{\emptyset}{|\atom{\Pi}|}=X$.
To this end,
assume that 
$\bigcup_{1\leq j\leq k}
(\tpi{\reduct{\Pi}{X}}{X\setminus C_j}{|C_j|}\cap C_j)=X$
and that $C_1,\dots,C_k$ are topologically ordered
such that the dependency graph of~$\Pi$ does not contain any edge
from an atom in $C_j$ to atoms in $C_{j+1}\cup\dots\cup C_k$ for $1\leq j\leq k$.
That is, atoms in $C_{j+1}\cup\dots\cup C_k$ do not occur
in rules $r\in\reduct{\Pi}{X}$ such that $\head{r}\in C_j$,
so that 
$\tpi{\reduct{\Pi}{X}}{X\setminus C_j}{|C_j|}\cap C_j=
 \tpi{\reduct{\Pi}{X}}{X\cap (C_1\cup\dots\cup C_{j-1})}{|C_j|}\cap C_j$
for $1\leq j\leq k$.
From
$X
 =\linebreak[1]
 (X\cap C_1)\cup\dots\cup(X\cap C_k)
 =
 (\tpi{\reduct{\Pi}{X}}{\emptyset}{|C_1|}\cap C_1)\cup\dots\cup
 (\tpi{\reduct{\Pi}{X}}{X\cap (C_1\cup\dots\cup C_{k-1})}{|C_k|}\cap C_k)
 \subseteq
 \tpi{\reduct{\Pi}{X}}{\emptyset}{|C_1|+\dots+|C_k|}
 =
 \tpi{\reduct{\Pi}{X}}{\emptyset}{|\atom{\Pi}|}
 \subseteq 
 X$,
we then conclude that $\tpi{\reduct{\Pi}{X}}{\emptyset}{|\atom{\Pi}|}=X$.
\end{proof}

For illustration, reconsider $\Pi_0$ in (\ref{ex:one})--(\ref{ex:tri})
and its dependency graph, looking as follows:
\begin{center}
\setlength{\unitlength}{1.2pt}
 \begin{picture}(90,34)(0,-14)
 \put(1,0){$p$}
 \put(90,0){$t$}
 \put(45.5,0){$r$}
 \put(45.5,-17){$s$}
 \put(45.5,17){$q$}
 \qbezier(8,3)(25,9)(42,3) \put(43,2.6){\vector(3,-1){1}}
 \qbezier(43,-1)(26,-7)(9,-1) \put(8,-0.7){\vector(-3,1){1}}
 \qbezier(86.5,3)(69.5,9)(52.5,3) \put(87.5,2.6){\vector(3,-1){1}}
 \qbezier(54,-1)(71,-7)(88,-1) \put(53,-0.7){\vector(-3,1){1}} 
 \put(7,-2.5){\vector(3,-1){37}} 
 \put(88.5,-2.5){\vector(-3,-1){37}} 
 \put(43.5,17.5){\vector(-3,-1){37}} 
 \put(53,17.5){\vector(3,-1){37}} 
 \end{picture}
\end{center}
Observe that the sets $\{s\}$, $\{p,r,t\}$, and $\{q\}$ of atoms
belong to SCCs (ordered topologically).
Furthermore, let us take the following reducts into account:
\begin{align*}
\reduct{\Pi_0}{\{p,r\}} =
\left\{
\begin{array}{@{}r@{}c@{}l@{}}
p & {}\leftarrow{} & 0\ \summ [ r=1, s=1 ].
\\
r & {}\leftarrow{} & 1\ \summ [ p=1, t=1 ].
\end{array}
\right\}
&&
\reduct{\Pi_0}{\{r,t\}} =
\left\{
\begin{array}{@{}r@{}c@{}l@{}}
t & {}\leftarrow{} & 1\ \summ [ r=1, s=1 ].
\\
r & {}\leftarrow{} & 1\ \summ [ p=1, t=1 ].
\end{array}
\right\}
\end{align*}
We have that
$\tpi{\reduct{\Pi_0}{\{p,r\}}}{\emptyset}{1}=\{p\}$ and
$\tpi{\reduct{\Pi_0}{\{p,r\}}}{\emptyset}{2}=
 \{p\}\cup\tps{\reduct{\Pi_0}{\{p,r\}}}{\{p\}}=
 \{p,r\}=
 \tpi{\reduct{\Pi_0}{\{p,r\}}}{\emptyset}{3}$.
Along with
$\tpi{\reduct{\Pi_0}{\{p,r\}}}{\{p,r\}}{1}\cap\{s\}=
 \tpi{\reduct{\Pi_0}{\{p,r\}}}{\{p,r\}}{1}\cap\{q\}=\linebreak[1]
 \emptyset$,
we obtain 
$(\tpi{\reduct{\Pi_0}{\{p,r\}}}{\{p,r\}}{1}\cap\{s\})\cup
 (\tpi{\reduct{\Pi_0}{\{p,r\}}}{\emptyset}{3}\cap\{p,r,t\})\cup
 (\tpi{\reduct{\Pi_0}{\{p,r\}}}{\{p,r\}}{1}\cap\{q\})=\linebreak[1]\{p,r\}$.
In view of Proposition~\ref{prop:tp},
this confirms that the model~$\{p,r\}$ of~$\Pi_0$
is an answer set of~$\Pi_0$.
On the other hand,
$(\tpi{\reduct{\Pi_0}{\{r,t\}}}{\{r,t\}}{1}\cap\{s\})\cup
 (\tpi{\reduct{\Pi_0}{\{r,t\}}}{\emptyset}{3}\cap\{p,r,t\})\cup
 (\tpi{\reduct{\Pi_0}{\{r,t\}}}{\{r,t\}}{1}\cap\{q\})=\emptyset$ 
yields that the model~$\{r,t\}$ of~$\Pi_0$
is not an answer set of~$\Pi_0$.}{}

\EXT{%
In a nutshell,
our encoding of answer sets (as counterexamples)
in disjunctive ASP combines the following parts:}%
{The above property is used in
 our encoding of answer sets (as counterexamples) in disjunctive ASP.
 In a nutshell, it combines the following parts:}
\begin{enumerate}
\item guessing an interpretation,
\item deriving the error-indicating atom \texttt{bot} if the interpretation is not a supported model
      (where each true atom occurs positively in the head of some rule whose body holds),
\item deriving \texttt{bot} if the true atoms of some \emph{non-trivial} SCC
      are not acyclicly derivable (checked via determining the complement of a fixpoint of~$\tp$), and
\item saturating interpretations that do not correspond to answer sets
      by deriving all truth assignments (for atoms) from \texttt{bot}.
\end{enumerate}
Note that the third part, checking acyclic derivability, concentrates
on atoms of non-trivial SCCs,
while checking support in the second part is already sufficient
for trivial SCCs.

\begin{lstlisting}[float=t,caption={Disjunctive meta-program (\texttt{\small metaD.lp}) for reified extended logic programs.},captionpos=b,label=fig:meta:disjunctive:one]
% extract supports of atoms and sums of weight lists' weights

supp(atom(A),B) :- rule(pos(atom(A)),   pos(B)).
supp(atom(A),B) :- rule(pos(sum(_,S,_)),pos(B)), wlist(S,_,pos(atom(A)),_).

sum(S,T)        :- elem(sum(_,S,_)), T = #sum [ wlist(S,Q,_,W) = W ].

% generate interpretation

true(atom(A)) | fail(atom(A)) :- elem(atom(A)).

fail(false).

true(conjunction(S)) :- elem(conjunction(S)),
                        true(P) : set(S,pos(P)), fail(N) : set(S,neg(N)).
fail(conjunction(S)) :- elem(conjunction(S)), set(S,pos(P)), fail(P).
fail(conjunction(S)) :- elem(conjunction(S)), set(S,neg(N)), true(N).

true(sum(L,S,U))     :- elem(sum(L,S,U)), sum(S,T),
                      L     #sum [ true(P) = W : wlist(S,Q,pos(P),W),
                                   fail(N) = W : wlist(S,Q,neg(N),W) ],
                      T-U   #sum [ fail(P) = W : wlist(S,Q,pos(P),W),
                                   true(N) = W : wlist(S,Q,neg(N),W) ].
fail(sum(L,S,U))     :- elem(sum(L,S,U)), sum(S,T),
                      T-L+1 #sum [ fail(P) = W : wlist(S,Q,pos(P),W),
                                   true(N) = W : wlist(S,Q,neg(N),W) ].
fail(sum(L,S,U))     :- elem(sum(L,S,U)),
                      U+1   #sum [ true(P) = W : wlist(S,Q,pos(P),W),
                                   fail(N) = W : wlist(S,Q,neg(N),W) ].

% verify supported model properties

bot :- rule(pos(H),pos(B)), true(B), fail(H).
bot :- true(atom(A)), fail(B) : supp(atom(A),B).

% verify acyclic derivability

step(C,Z)                :- scc(C,_), Z = #sum [ scc(C,pos(atom(A))) ].

sccw(A)                  :- scc(C,pos(atom(A))),
                            fail(B) : supp(atom(A),B) : not scc(C,pos(B)).

wait(E,D-1)              :- scc(C,pos(E)), fail(E), step(C,Z), D = 1..Z.

wait(atom(A),0)          :- scc(C,pos(atom(A))).
wait(atom(A),D)          :- scc(C,pos(atom(A))), sccw(A), step(C,Z), D = 1..Z,
                            wait(B,D-1) : supp(atom(A),B) : scc(C,pos(B)).

wait(sum(L,S,U),D-1) :- scc(C,pos(sum(L,S,U))), sum(S,T), step(C,Z), D = 1..Z,
       T-L+1 #sum [ fail(P) = W : wlist(S,Q,pos(P),W) : not scc(C,pos(P)),
                wait(P,D-1) = W : wlist(S,Q,pos(P),W) :     scc(C,pos(P)),
                    true(N) = W : wlist(S,Q,neg(N),W) ].

wait(conjunction(S),D-1) :- scc(C,pos(conjunction(S))), set(S,pos(P)),
                            scc(C,pos(P)), wait(P,D-1), step(C,Z), D = 1..Z.

bot :- scc(C,pos(atom(A))), true(atom(A)), wait(atom(A),Z), step(C,Z).

% saturate interpretations that are not answer sets

true(atom(A)) :- elem(atom(A)), bot.
fail(atom(A)) :- elem(atom(A)), bot.
\end{lstlisting}
The meta-program in Listing~\ref{fig:meta:disjunctive:one} implements the sketched idea.
In the following, we concentrate on describing its crucial features.
For evaluating support, the meta-rules in Line~3 and~4
collect atoms having a positive occurrence in the head of a rule
along with the rule's body.
Note that, for atoms contained in a $\summ$ constraint in the head,
the associated bounds and weights are inessential in the context of support.
On the other hand, the meta-rule in Line~6 sums the weights of all literals in a $\summ$ constraint;
this is needed to evaluate bounds in the sequel,
where (non-reified) default negation and upper bounds (acting negatively)
are inapplicable in view of saturation.

The meta-rules in Line~10--29 generate an interpretation by guessing 
some truth value for each atom (Line~10) and
evaluating further constructs occurring in a reified program accordingly (Line~12--29).
While the special constant \texttt{false} (used as head of integrity constraints)
holds in no interpretation (\texttt{fail(false)} is a fact)
and the evaluation of conjunctions is straightforward,
more care is required for evaluating $\summ$ constraints.
For instance, the case that a $\summ$ constraint holds is
in the meta-rule in Line~19--23 identified via sufficiently
many literals that hold to achieve the lower bound~\texttt{L} and also
sufficiently many literals that do not hold to fill the gap
between the upper bound~\texttt{U} and the sum~\texttt{T} of all weights.
Note that the latter condition is encoded by the lower bound \texttt{T-U},
rather than taking~\texttt{U} as an upper bound (as done in \texttt{meta.lp}).
The complementary cases that a $\summ$ constraint does not hold
are described in the same manner in Line~24--29,
where the lower bound \texttt{T-L+1} (or \texttt{U+1}) for weights of
literals that do not hold (or hold) is used to indicate a violated lower (or upper)
bound of the reified $\summ$ constraint.

Given an interpretation of atoms and the corresponding truth values 
of further constructs in an extended program,
the meta-rules in Line~33 and~34 are used to derive~\texttt{bot}
if the interpretation does not provide us with a supported model.
To avoid such a derivation of~\texttt{bot}, every rule of the
reified program must be satisfied, and every true atom must have
a positive occurrence in the head of some rule whose body holds.

It remains to check the acyclic derivability of atoms belonging to
non-trivial SCCs.
To this end,
the meta-rule in Line~38 determines the number~\texttt{Z} of atoms in an SCC
labeled~\texttt{C} as the maximum step at which a fixpoint of~$\tp$,
applied locally to~\texttt{C}, is reached.
Furthermore, the meta-rule in Line~40--41 derives~\texttt{sccw(A)} if
the atom referred to by~\texttt{A} does not have a derivation external to~\texttt{C}.
(Recall that the positive body elements of rules internally connecting an SCC,
 i.e., rules contributing the SCC's edges to the dependency graph,
 are marked by facts over \texttt{scc/2}; cf.\ Listing~\ref{fig:reified}.)
The acyclic derivability of atoms indicated by \texttt{sccw(A)} is of
particular interest in the sequel.
In fact,
our encoding identifies the complement of a fixpoint of~$\tp$ 
in terms of atoms~\texttt{A} for which \texttt{wait(atom(A),Z)} is derived.
To accomplish this,
the meta-rule in Line~45 marks all atoms of~\texttt{C} as underived at step~\texttt{0}.
As encoded via the meta-rule in Line~46--47,
an atom~\texttt{A} stays underived at a later step~\texttt{D}
if there is no external derivation of~\texttt{A} (\texttt{sccw(A)} holds) and
the bodies~\texttt{B} of all component-internal supports of~\texttt{A}
are yet underived at step \texttt{D-1} (\texttt{wait(B,D-1)} holds).
The latter is checked via the meta-rules in Line~49--52 and~54--55, respectively.
The former applies to $\summ$ constraints
and identifies cases where the weights of literals that do not hold 
along with the ones of yet underived atoms of~\texttt{C} exceed~\texttt{T-L},
so that the lower bound~\texttt{L} is not yet established.
Similarly, the underivability of a conjunction is recognized 
via a yet underived positive body element internal to the component~\texttt{C}.
Also note that the falsity of elements of~\texttt{C} 
is propagated via the meta-rule in Line~43,
so that false atoms, $\summ$ constraints, and conjunctions
do not contribute to derivations of atoms of~\texttt{C}.
As mentioned above, 
the complement of a fixpoint of~$\tp$ contains the
atoms~\texttt{A} such that \texttt{wait(atom(A),Z)} is eventually derived.
If any such atom~\texttt{A} is true,
failure to construct an answer set is indicated by
deriving~\texttt{bot} via the meta-rule in Line~57.

\EXT{%
For illustration,
consider the ground rules 
shown in Listing~\ref{fig:ground:scc},
which are 
obtained for the SCC labeled~\texttt{0}
in Listing~\ref{fig:reified}.
\lstinputlisting[float=t,caption={(Simplified) ground rules obtained from meta-rules in Line~38--57 of \texttt{\small metaD.lp}.},captionpos=b,label=fig:ground:scc]{toySCC.lp}
%
For the answer set $\{p,r\}$ of~$\Pi_0$,
represented by the atoms
\texttt{true(atom(p))}, \texttt{true(atom(r))},
\texttt{fail(atom(q))}, \texttt{fail(atom(s))}, and \texttt{fail(atom(t))},
we have that 
\texttt{wait(sum(1,1,2),\linebreak[1]0;1;2)} and 
\texttt{wait(conjunction(0),\linebreak[1]0;1;2)}
are underivable via the rules in Line~26--28 and~34--36, respectively.
Thus,
\texttt{wait(atom(p),\linebreak[1]1;2;3)}
are not derived via the rules in Line~12--14,
so that 
\texttt{wait(sum(1,0,2),\linebreak[1]1;2)} and 
\texttt{wait(conjunction(1),\linebreak[1]1;2)}
are underivable in turn via the rules in Line~31--32 and 39--40, respectively.
As a consequence,
\texttt{wait(atom(r),\linebreak[1]2;3)}
are not derived via the rules in Line~18--19.
We have thus checked that none of the rules in Line~42--44 allows for deriving \texttt{bot},
which tells us that the true atoms~$p$ and~$r$ are acyclicly derivable.
On the other hand, for the interpretation $\{r,t\}$,
given by 
\texttt{true(atom(r))}, \texttt{true(atom(t))},
\texttt{fail(atom(p))}, \texttt{fail(atom(q))}, and \texttt{fail(atom(s))},
we have that the atoms
\texttt{wait(atom(p),\linebreak[1]0;1;2)},
\texttt{wait(sum(1,1,2),\linebreak[1]0;1;2)}, 
\texttt{wait(sum(1,0,2),\linebreak[1]0;1;2)},
\texttt{wait(conjunction(0),\linebreak[1]0;1;2)},
\texttt{wait(conjunction(1),\linebreak[1]0;1;2)},
\texttt{wait(atom(r),\linebreak[1]0;1;2;3)}, and 
\texttt{wait(atom(t),\linebreak[1]0;1;2;3)}
are derived in turn via the rules in Line~3 and 16--40.
From
\texttt{true(atom(r))} and \texttt{wait(atom(r),3)} as well as
\texttt{true(atom(t))} and \texttt{wait(atom(t),3)},
we further derive the error-indicating atom \texttt{bot}
via the rules in Line~43 and~44.
This signals that the true atoms $r$ and~$t$ are \emph{not} acyclicly derivable,
so that $\{r,t\}$ is not an answer set of~$\Pi_0$.}%
{}

Finally, saturation of interpretations that do not correspond to answer sets
is accomplished via the meta-rules in Line~61 and~62 of Listing~\ref{fig:meta:disjunctive:one}.
They make sure that~\texttt{bot} is included in an answer set of the meta-program
only if it is inevitable wrt every interpretation.
When considering the encoding part in Listing~\ref{fig:meta:disjunctive:one}
in isolation, it like \texttt{meta.lp} describes answer sets of a reified program,
and \texttt{bot} is derived only if there is no such answer set.

\EXT{%
\begin{lstlisting}[float=t,caption={Meta-program for complex optimization (\texttt{\small metaO.lp}) on reified logic programs.},captionpos=b,label=fig:meta:optimize:one,linerange={1-64}]
% extract (complex) optimization criteria per priority level and weight
% (relative to user predicate optimize/3; cardinality taken by default;
%  Pareto improvement over weights used for comparison at a priority level)

cxopt(card).  cxopt(incl).  cxopt(pref).

cxopt(J,W,O)    :- minimize(J,S), wlist(S,_,_,W), cxopt(O), optimize(J,W,O).
cxopt(J,W,card) :- minimize(J,S), wlist(S,_,_,W),
                   not optimize(J,W,O) : optimize(J,W,O).

% verify dominance

equal(J)     :- cxopt(J,_,_), equal(J,W,O) : cxopt(J,W,O).

chain(J1,J2) :- cxopt(J1;J2,_,_), J2 < J1,
                not cxopt(J3,W,O) : cxopt(J3,W,O) : J2 < J3 : J3 < J1.

check(J2)    :- cxopt(J2,_,_), not chain(J1,J2) : chain(J1,J2).
check(J2)    :- chain(J1,J2), check(J1), equal(J1).

bot          :- not cxopt(J,W,O) : cxopt(J,W,O).
bot          :- check(J1), worse(J1).
bot          :- check(J1), equal(J1), not chain(J1,J2) : chain(J1,J2).

% require non-existence of dominating answer set

:- not bot.

% check cardinality criteria

index(S,Q)         :- cxopt(J,_,card), minimize(J,S), wlist(S,Q,_,_),
                      not wlist(S,Q+1,E,W) : wlist(S,Q+1,E,W).

count(S,W,-1, 0)   :- cxopt(J,W,card), minimize(J,S).
count(S,W,Q+1,I)   :- count(S,W,Q,I), wlist(S,Q+1,_,_).
count(S,W,Q+1,I+1) :- count(S,W,Q,I), wlist(S,Q+1,pos(P),W),     hold(P).
count(S,W,Q+1,I+1) :- count(S,W,Q,I), wlist(S,Q+1,neg(N),W), not hold(N).

cdown(S,W,Q,  I)   :- count(S,W,Q,I), index(S,Q), not count(S,W,Q,I+1).
cdown(S,W,Q-1,I)   :- cdown(S,W,Q,I), wlist(S,Q,_,_).
cdown(S,W,Q-1,I-1) :- cdown(S,W,Q,I), wlist(S,Q,pos(P),W), true(P), 0 <= I.
cdown(S,W,Q-1,I-1) :- cdown(S,W,Q,I), wlist(S,Q,neg(N),W), fail(N), 0 <= I.

equal(J,W,card)    :- cxopt(J,W,card), minimize(J,S), cdown(S,W,-1,0).

worse(J)           :- cxopt(J,W,card), minimize(J,S), cdown(S,W,-1,-1).

% check inclusion criteria

ndiff(pos(P))   :- cxopt(J,W,incl), minimize(J,S), wlist(S,_,pos(P),W),
                   true(P).
ndiff(pos(P))   :- cxopt(J,W,incl), minimize(J,S), wlist(S,_,pos(P),W),
                            not hold(P).
ndiff(neg(N))   :- cxopt(J,W,incl), minimize(J,S), wlist(S,_,neg(N),W),
                   fail(N).
ndiff(neg(N))   :- cxopt(J,W,incl), minimize(J,S), wlist(S,_,neg(N),W),
                                hold(N).

equal(J,W,incl) :- cxopt(J,W,incl), minimize(J,S), ndiff(E) : wlist(S,_,E,W).

worse(J)        :- cxopt(J,W,incl), minimize(J,S), wlist(S,_,pos(P),W),
                   true(P), not hold(P).
worse(J)        :- cxopt(J,W,incl), minimize(J,S), wlist(S,_,neg(N),W),
                   fail(N),     hold(N).
\end{lstlisting}
}%
{%
\begin{lstlisting}[float=t,caption={Meta-program for complex optimization (\texttt{\small metaO.lp}) on reified logic programs.},captionpos=b,label=fig:meta:optimize:one,linerange={1-27,47-64}]
% extract (complex) optimization criteria per priority level and weight
% (relative to user predicate optimize/3; cardinality taken by default;
%  Pareto improvement over weights used for comparison at a priority level)

cxopt(card).  cxopt(incl).  cxopt(pref).

cxopt(J,W,O)    :- minimize(J,S), wlist(S,_,_,W), cxopt(O), optimize(J,W,O).
cxopt(J,W,card) :- minimize(J,S), wlist(S,_,_,W),
                   not optimize(J,W,O) : optimize(J,W,O).

% verify dominance

equal(J)     :- cxopt(J,_,_), equal(J,W,O) : cxopt(J,W,O).

chain(J1,J2) :- cxopt(J1;J2,_,_), J2 < J1,
                not cxopt(J3,W,O) : cxopt(J3,W,O) : J2 < J3 : J3 < J1.

check(J2)    :- cxopt(J2,_,_), not chain(J1,J2) : chain(J1,J2).
check(J2)    :- chain(J1,J2), check(J1), equal(J1).

bot          :- not cxopt(J,W,O) : cxopt(J,W,O).
bot          :- check(J1), worse(J1).
bot          :- check(J1), equal(J1), not chain(J1,J2) : chain(J1,J2).

% require non-existence of dominating answer set

:- not bot.

% check inclusion criteria

ndiff(pos(P))   :- cxopt(J,W,incl), minimize(J,S), wlist(S,_,pos(P),W),
                   true(P).
ndiff(pos(P))   :- cxopt(J,W,incl), minimize(J,S), wlist(S,_,pos(P),W),
                            not hold(P).
ndiff(neg(N))   :- cxopt(J,W,incl), minimize(J,S), wlist(S,_,neg(N),W),
                   fail(N).
ndiff(neg(N))   :- cxopt(J,W,incl), minimize(J,S), wlist(S,_,neg(N),W),
                                hold(N).

equal(J,W,incl) :- cxopt(J,W,incl), minimize(J,S), ndiff(E) : wlist(S,_,E,W).

worse(J)        :- cxopt(J,W,incl), minimize(J,S), wlist(S,_,pos(P),W),
                   true(P), not hold(P).
worse(J)        :- cxopt(J,W,incl), minimize(J,S), wlist(S,_,neg(N),W),
                   fail(N),     hold(N).
\end{lstlisting}
}%
Our meta-programs \texttt{meta.lp} and \texttt{metaD.lp}
in Listing~\ref{fig:meta} and~\ref{fig:meta:disjunctive:one}
have not yet considered facts \texttt{minimize(J,S)} in reified programs,
reflecting input $\minimize$ statements.
In fact, complex optimization is addressed by the meta-program \texttt{metaO.lp},
\EXT{whose first part is }{}shown in Listing~\ref{fig:meta:optimize:one}.
It allows for separate optimization criteria per priority level~\texttt{J} and
weight~\texttt{W} (in facts \texttt{wlist(S,Q,E,W)}).
Particular criteria can be provided via the user predicate \texttt{optimize(J,W,O)},
where the values \texttt{card}, \texttt{incl}, and \texttt{pref} for~\texttt{O}
refer to minimality regarding cardinality, inclusion, and 
preference \cite{sakino00a}, respectively, among the involved
literals~\texttt{E}. 
Such criteria are reflected via instances of \texttt{cxopt(J,W,O)}, 
derived via the rules in Line~7 and~8--9,
where \texttt{card} is taken by default if no criterion is provided by the user.
At each priority level~\texttt{J}, Pareto improvement of a counterexample
(constructed via the rules in \texttt{metaD.lp})
over all weights~\texttt{W} and criteria~\texttt{O} 
such that \texttt{cxopt(J,W,O)} holds
is used for deciding whether a candidate answer set
(constructed via the rules in \texttt{meta.lp}) is optimal.
To this end,
similarity 
at a priority level~\texttt{J} is indicated by deriving \texttt{equal(J)}
from \texttt{equal(J,W,O)} over all instances of \texttt{cxopt(J,W,O)}
via the rule in Line~13.
Furthermore, the rules in Line~15--19 are used to chain successive priority levels,
where a greater level~\texttt{J1} is more significant than its smaller neighbor~\texttt{J2},
and to signal whether a priority level~\texttt{J2} is taken into account.
The latter is the case if \texttt{equal(J1)} has been derived at all more significant
priority levels~\texttt{J1}.
If it turns out that a candidate answer set is not refuted by a dominating counterexample,
we derive \texttt{bot} via the rules in Line~21, 22, and~23:
the first rule applies if there are no optimization criteria at all,
the second one checks whether the counterexample is worse (or incomparable),
as indicated by \texttt{worse(J1)}
at an inspected priority level~\texttt{J1},
and the third one detects lack of Pareto improvement
from equality at the lowest priority level.
Finally, the integrity constraint in Line~27 stipulates \texttt{bot} to hold.
Along with saturation (in \texttt{metaD.lp}),
this implies that a candidate answer set (constructed via the rules in \texttt{meta.lp})
is accepted only if there is no dominating counterexample,
thus selecting exactly the optimal answer sets of an input program.
The described rules serve the general purpose of identifying undominated answer sets,
and the remainder of \texttt{metaO.lp} defines
\texttt{equal(J,W,O)} and \texttt{worse(J)} relative to particular
optimization criteria.

\EXT{%
The rules in Line~31--46 implement the comparison of cardinalities 
between a candidate answer set and a counterexample at a priority level~\texttt{J}
over the literals indicated by a weight~\texttt{W} (in facts \texttt{wlist(S,Q,E,W)}).
To this end, we first count the number of such literals~\texttt{E}
that hold wrt the candidate answer set via the rules in Line~34--37.
The outcome, i.e., the greatest number~\texttt{I} such that 
\texttt{count(S,W,Q,I)} is derived at the maximum index~\texttt{Q}
of the list~\texttt{S} given in \texttt{minimize(J,S)},
is identified via the rule in Line~39 and taken as the starting value
for counting down the literals that hold wrt the counterexample.
Note that this approach omits a (quadratic) comparison between the precise
cardinalities obtained wrt the candidate answer set and the counterexample.
The result of counting down, accomplished via the rules in Line~40--42,
is given by the derived instances of \texttt{cdown(S,W,-1,I)}, where two
outcomes deserve particular attention.
In fact, \texttt{cdown(S,W,-1,0)} indicates that at least as many literals
as wrt the candidate answer set hold wrt the counterexample, so that
\texttt{equal(J,W,card)} is derived via the rule in Line~44.
On the other hand, \texttt{cdown(S,W,-1,-1)} signals that more literals hold
wrt the counterexample;
as this makes a Pareto improvement at priority level~\texttt{J} impossible,
we derive \texttt{worse(J)} via the rule in Line~46.
Finally, note that we did not directly encode minimization of cost functions with coefficients
different from~$1$.
In view of the multiset semantics supported by \gringo, such coefficients
can still be represented by including the desired number of duplicates of
a literal in an input $\minimize$ statement.

The second optimization criterion, inclusion (indicated via \texttt{cxopt(J,W,incl)}),
is implemented by the rules in Line~50--64.
}%
{Inclusion-based minimization, indicated via \texttt{cxopt(J,W,incl)},
 is implemented by the rules in Line~31--45.
}%
The test for equality, attested by deriving \texttt{equal(J,W,incl)} via the rule in Line~\EXT{59}{40},
is accomplished by checking whether a candidate answer set and a (comparable) counterexample agree 
on all involved literals~\texttt{E}; otherwise, \texttt{ndiff(E)} is not derived
via the rules in Line~\EXT{50--57}{31--38}. 
Furthermore,
the counterexample is incomparable to the candidate answer set
if it includes some literal not shared by the latter;
in such a case, \texttt{worse(J)} is derived via the rules in Line~\EXT{61--62}{42--43} and~\EXT{63--64}{44--45}.
In fact, the three $\subseteq^1_1$-minimal answer sets
of $\Pi_1$ (given in \texttt{p1.lp}), consisting of the rules in (\ref{ex:one})--(\ref{ex:tri})
and the $\minimize$ statement in~(\ref{ex:for})
can now be computed in the following way:
\begin{lstlisting}[numbers=none,basicstyle=\ttfamily,aboveskip=\smallskipamount,belowskip=\smallskipamount,xleftmargin=\parindent]
gringo --reify p1.lp | gringo meta.lp metaD.lp metaO.lp \
  <(echo "optimize(1,1,incl).") - | claspD 0
\end{lstlisting}
Observe that \claspD\ \cite{drgegrkakoossc08a}, the disjunctive extension of
\clasp\ \cite{gekanesc07a}, 
is used for solving the proper disjunctive ground program obtained from \gringo.

\EXT{%
\begin{lstlisting}[float=t,caption={Meta-program for literal preferences (\texttt{\small metaO.lp}) on reified logic programs.},captionpos=b,label=fig:meta:optimize:two,linerange={66-106},firstnumber=66]
% check preference criteria (relative to user predicate prefer/2)

cando(pos(P))   :- cxopt(J,W,pref), minimize(J,S), wlist(S,_,pos(P),W),
                   fail(P),     hold(P).
cando(neg(N))   :- cxopt(J,W,pref), minimize(J,S), wlist(S,_,neg(N),W),
                   true(N), not hold(N).
nocan(pos(P))   :- cxopt(J,W,pref), minimize(J,S), wlist(S,_,pos(P),W),
                   true(P).
nocan(pos(P))   :- cxopt(J,W,pref), minimize(J,S), wlist(S,_,pos(P),W),
                            not hold(P).
nocan(neg(N))   :- cxopt(J,W,pref), minimize(J,S), wlist(S,_,neg(N),W),
                   fail(N).
nocan(neg(N))   :- cxopt(J,W,pref), minimize(J,S), wlist(S,_,neg(N),W),
                                hold(N).

condo(pos(P))   :- cxopt(J,W,pref), minimize(J,S), wlist(S,_,pos(P),W),
                   true(P), not hold(P).
condo(neg(N))   :- cxopt(J,W,pref), minimize(J,S), wlist(S,_,neg(N),W),
                   fail(N),     hold(N).
nocon(pos(P))   :- cxopt(J,W,pref), minimize(J,S), wlist(S,_,pos(P),W),
                   fail(P).
nocon(pos(P))   :- cxopt(J,W,pref), minimize(J,S), wlist(S,_,pos(P),W),
                                hold(P).
nocon(neg(N))   :- cxopt(J,W,pref), minimize(J,S), wlist(S,_,neg(N),W),
                   true(N).
nocon(neg(N))   :- cxopt(J,W,pref), minimize(J,S), wlist(S,_,neg(N),W),
                            not hold(N).

cando(S,W,E)    :- cxopt(J,W,pref), minimize(J,S), wlist(S,_,E;E1,W),
                   prefer(E,E1), E1 != E, cando(E), condo(E1).

nocon(S,W,E)    :- cxopt(J,W,pref), minimize(J,S), wlist(S,_,E,W), nocon(E).
nocon(S,W,E)    :- cxopt(J,W,pref), minimize(J,S), wlist(S,_,E,W),
                   nocan(E1) : wlist(S,_,E1,W) : E1 != E : prefer(E,E1).
nocon(S,W,E)    :- cxopt(J,W,pref), minimize(J,S), wlist(S,_,E;E2,W),
                   prefer(E2,E), not prefer(E,E2), cando(E2).

equal(J,W,pref) :- cxopt(J,W,pref), minimize(J,S), cando(S,W,E),
                   nocon(E2) : wlist(S,_,E2,W) : prefer(E2,E) : not prefer(E,E2).

worse(J)        :- cxopt(J,W,pref), minimize(J,S), nocon(S,W,E) : wlist(S,_,E,W).
\end{lstlisting}
The third optimization criterion implemented in \texttt{metaO.lp}
(Listing~\ref{fig:meta:optimize:two})
relies on literal preferences according to \cite{sakino00a},
i.e., a relation $\ell_1\preceq \ell_2$ provided via the user predicate \texttt{prefer/2}.
In a nutshell, an answer set $X_1$ is preferable to another answer set $X_2$, i.e., $X_1\preceq X_2$,
if there is a preference $\ell_1\preceq \ell_2$ such that
$\ell_1$ holds wrt~$X_1$, but not wrt~$X_2$, and
$\ell_2$ holds wrt~$X_2$, but not wrt~$X_1$,
while $\ell_2'\npreceq \ell_1$ or $\ell_1\preceq \ell_2'$
applies to every literal $\ell_2'$ that holds wrt~$X_2$, but not wrt~$X_1$.
The idea of the encoding part in Listing~\ref{fig:meta:optimize:two} is to derive
\texttt{equal(J,W,pref)} if 
a candidate answer set~$X_1$ is preferable to a counterexample~$X_2$,
and \texttt{worse(J)} if $X_2$ is not preferable to~$X_1$.
If both \texttt{equal(J,W,pref)} and \texttt{worse(J)} are underivable,
it shows that $X_2$ is preferable to~$X_1$, but not vice versa,
so that~$X_1$ is dominated by~$X_2$ wrt the projection of the
preference relation specified via \texttt{prefer/2} to literals~\texttt{E}
qualified by \texttt{minimize(J,S)} and \texttt{wlist(S,\_,E,W)}.
To implement these checks,
the rules in Line~68--92 identify literals at which a candidate answer
set~$X_1$ and a counterexample~$X_2$ differ.
The literals~\texttt{E} that hold wrt~$X_1$, but not wrt~$X_2$,
are indicated by \texttt{cando(E)}, while \texttt{nocan(E)} is derived otherwise.
Similarly, \texttt{condo(E)} expresses that \texttt{E} holds wrt~$X_2$,
but not wrt~$X_1$, whereas \texttt{nocon(E)} indicates that this is not the case.
The rule in Line~94--95 further derives \texttt{cando(S,W,E)}
if~\texttt{E} is a literal that holds wrt~$X_1$ only such
that \texttt{E}${}\preceq{}$\texttt{E1} for a literal~\texttt{E1}
that holds wrt~$X_2$ only.
If such a literal~\texttt{E} is not defeated by another literal~\texttt{E2} 
satisfying \texttt{E2}${}\preceq{}$\texttt{E} as well as \texttt{E}${}\npreceq{}$\texttt{E2}
and holding wrt~$X_2$ only,
we derive \texttt{equal(J,W,pref)} 
via the rule in Line~103--104.
On the other hand, a literal~\texttt{E} provides no evidence for
the counterexample~$X_2$ being preferable to the candidate answer set~$X_1$
if \texttt{E} does not hold wrt~$X_2$ only,
if no~\texttt{E1} such that \texttt{E}${}\preceq{}$\texttt{E1}
holds wrt~$X_1$ only, or if
some~\texttt{E2} such that \texttt{E2}${}\preceq{}$\texttt{E} and
\texttt{E}${}\npreceq{}$\texttt{E2} holds wrt~$X_1$ only.
If some of these cases applies,
\texttt{nocon(S,W,E)} is derived via the rules in Line~97--101.
Given this,
we further derive \texttt{worse(J)} via the rule in Line~106
if none of the literals~\texttt{E} qualified by \texttt{wlist(S,\_,E,W)}
witnesses that $X_2$ is preferable to~$X_1$.
Finally, recall that the counterexample~$X_2$ dominates the candidate
answer set~$X_1$ if~$X_2$ is preferable to~$X_1$, but not vice versa,
i.e., if neither \texttt{equal(J,W,pref)} nor \texttt{worse(J)} is derivable
in view of \texttt{cxopt(J,W,pref)}.}%
{In addition to inclusion-based minimization,
 \texttt{metaO.lp}     implements 
 comparisons wrt cardinality \cite{siniso02a} and 
 literal preferences \cite{sakino00a},
 activatable via facts of the form 
 \texttt{optimize(J,W,card)} and \texttt{optimize(J,W,pref)}
 (along with \texttt{prefer(E1,E2)}), respectively.
 For space reasons, the details are omitted here; they
 can be found in the extended version of this paper \cite{gekasc11A}.
 The latter also provides formal results and arguments demonstrating the 
 correctness of our meta-programming technique
 wrt the specification of optimal answer sets in the background.}

\EXT{%
For showing the correctness of our meta-programming approach,
we make use of the property that any answer set generated 
via \texttt{meta.lp} and \texttt{metaD.lp} encapsulates a pair of
answer sets of the reified input program.
\begin{lemma}\label{lem:as}
Let $F$ be the set of facts in the reification of
an extended logic program~$\Pi$.
\SPACY{}
Then, 
$X$ and~$Y$ are answer sets of~$\Pi$
iff
there is a unique answer set $Z$ of 
$F\cup
  \{\text{\texttt{meta.lp}}\}
  \cup
  \{\text{\texttt{metaD.lp}}\}$
such that
$X=\{a \mid \text{\texttt{hold(atom($a$))}}\in Z\}$ and
$Y=\{a \mid \text{\texttt{true(atom($a$))}}\in Z\}$.
\end{lemma}
\begin{proof}[Proof (Sketch)]
Given that the meta-rules in \texttt{meta.lp} merely
reconstruct $\Pi$ from its reification,
we have that $X$ is an answer set of~$\Pi$ iff
there is a unique answer set~$X'$ of $F\cup\{\text{\texttt{meta.lp}}\}$
such that
$X=\{a \mid \text{\texttt{hold(atom($a$))}}\in X'\}$,
where $X'$ additionally includes $F$,
atoms over \texttt{hold/1} for body elements of~$\Pi$ that hold wrt~$X$,
as well as auxiliary atoms over \texttt{litb/1}, \texttt{eleb/1}, and \texttt{elem/1}.
Along with the fact that the (ground instances of the) rules in \texttt{metaD.lp}
are invariant under reduct,
do not include any integrity constraint,
and do not define any predicates occurring in $F\cup\{\text{\texttt{meta.lp}}\}$,
any answer set~$X'$ of $F\cup\{\text{\texttt{meta.lp}}\}$
can be extended to an answer set $Z=X'\cup Y'$ of 
$F\cup
  \{\text{\texttt{meta.lp}}\}
  \cup
  \{\text{\texttt{metaD.lp}}\}$,
where $Y'$ is a minimal model of~$F$, 
the auxiliary rules in Line~3--13 of \texttt{meta.lp},
and \texttt{metaD.lp}.
In view of the (proper) disjunctive rule in Line~10 of \texttt{metaD.lp}, 
for each $a\in\atom{\Pi}$, we have that
$\text{\texttt{true(atom($a$))}}\in\nolinebreak Y'$ or
$\text{\texttt{fail(atom($a$))}}\in Y'$,
while saturation in Line~61--62 of \texttt{metaD.lp} admits
$\text{\texttt{bot}}\in Y'$ 
only if \texttt{bot} 
is necessarily derived via the rules in Line~33, 34, and~57 of \texttt{metaD.lp}.
As a consequence, the rules in \texttt{metaD.lp} check for the existence of a model~$Y'$ such that 
$Y=\{a \mid \text{\texttt{true(atom($a$))}}\in\nolinebreak Y'\}$ is an answer set of~$\Pi$.
In particular, the rules in Line~33 and~34 of \texttt{metaD.lp}
derive \texttt{bot} wrt interpretations that are not models of~$\Pi$ or
include some atom that does not occur in the head of any rule of~$\Pi$ whose body holds.
Moreover, the rule in Line~57 of \texttt{metaD.lp}
derives \texttt{bot} wrt interpretations containing some atom of a non-trivial SCC of the
dependency graph of~$\Pi$ that does not belong to the least fixpoint of~$\tp$
applied to the SCC.
In view of Proposition~\ref{prop:tp},
we conclude that $Y'$ is a minimal model (more precisely, an answer set) of~$F$, 
the rules in Line~3--13 of \texttt{meta.lp},
and \texttt{metaD.lp} such that $\text{\texttt{bot}}\notin Y'$ iff
$Y$ is an answer set of~$\Pi$.
Hence,
we have that 
$X$ and~$Y$ are answer sets of~$\Pi$
iff
there is an answer set~$Z$  of
$F\cup
  \{\text{\texttt{meta.lp}}\}
  \cup
  \{\text{\texttt{metaD.lp}}\}$
such that
$X=\{a \mid \text{\texttt{hold(atom($a$))}}\in\nolinebreak Z\}$ and
$Y=\{a \mid \text{\texttt{true(atom($a$))}}\in Z\}$.
Finally,
any such answer set~$Z$ of 
$F\cup
  \{\text{\texttt{meta.lp}}\}
  \cup
  \{\text{\texttt{metaD.lp}}\}$
is unique since all atoms in~$Z$
that are not of the form \texttt{hold(atom($a$))}
or \texttt{true(atom($a$))} follow deterministically
once the interpretation of these atoms is fixed.
\end{proof}

We are now ready to show our main correctness result
wrt the specification of optimal answer sets given in the background.
\begin{theorem}\label{thm:sound}
Let $F$ be the set of facts in the reification of
an extended logic program~$\Pi$ and a statement
$\minimize[\ell_1=w_1@J_1,\dots,\ell_k=w_k@J_k]$
such that $\{\ell_1,\dots,\ell_k\}\subseteq\atom{\Pi}\cup\{\naf{a}\mid a\in\atom{\Pi}\}$,
$M\subseteq\{\leq^{w_i}_{J_i},\subseteq^{w_i}_{J_i},{\preceq^{w_i}_{J_i}} \mid 1\leq\nolinebreak i\leq\nolinebreak k\}$,
and 
${\preceq}\subseteq(\atom{\Pi}\cup\{\naf{a}\mid a\in\atom{\Pi}\})\times(\atom{\Pi}\cup\{\naf{a}\mid a\in\atom{\Pi}\})$. 
\SPACY{}
Then, 
$X$ is an optimal answer set of~$\Pi$ wrt~$M$
iff
there is a unique answer set $Z$ of 
\begin{align*}
\Pi_{\text{meta}}={}
& F
  \cup
  \{\text{\texttt{meta.lp}}\}
  \cup
  \{\text{\texttt{metaD.lp}}\}\cup\{\text{\texttt{metaO.lp}} 
  \}
\\ {} \cup {} & 
  \{\text{\texttt{optimize($J$,$w$,card).}} \mid {\leq^w_J} \in M\}
\\ {} \cup {} & 
  \{\text{\texttt{optimize($J$,$w$,incl).}} \mid {\subseteq^w_J} \in M\}
\\ {} \cup {} & 
  \{\text{\texttt{optimize($J$,$w$,pref).}} \mid {\preceq^w_J} \in M\}
\\ {} \cup {} & 
  \{\text{\texttt{prefer($\ell_1'$,$\ell_2'$).}} \mid \ell_1 \preceq \ell_2\}
\end{align*}
such that $X=\{a \mid \text{\texttt{hold(atom($a$))}}\in Z\}$,
where $\ell'=\text{\texttt{pos(atom($\ell$))}}$ if $\ell\in\atom{\Pi}$
and, respectively, $\ell'=\text{\texttt{neg(atom($a$))}}$ if 
$\ell=\naf a$ for $a\in\atom{\Pi}$.
\end{theorem}
\begin{proof}[Proof (Sketch)]
By Lemma~\ref{lem:as},
we have that $X$ and $Y$ are answer sets of~$\Pi$
iff 
there is a (unique) answer set~$Z'$ of
$F\cup
  \{\text{\texttt{meta.lp}}\}
  \cup
  \{\text{\texttt{metaD.lp}}\}$
such that
$X=\{a \mid \text{\texttt{hold(atom($a$))}}\in Z'\}$ and
$Y=\{a \mid \text{\texttt{true(atom($a$))}}\in Z'\}$.
Since $\{\text{\texttt{metaO.lp}}\}\setminus\{\text{\texttt{:- not bot.}}\}$
is stratified,
there is a unique answer set $Z$ of
$Z'\cup(\{\text{\texttt{metaO.lp}}\}\setminus\{\text{\texttt{:- not bot.}}\})
 \cup \{\text{\texttt{optimize($J$,$w$,card).}} \mid {\leq^w_J} \in M\} \cup
 \{\text{\texttt{optimize($J$,$w$,incl).}} \mid {\subseteq^w_J} \in M\} \cup
 \{\text{\texttt{optimize($J$,$w$,pref).}} \mid {\preceq^w_J} \in M\} \cup
 \{\text{\texttt{prefer($\ell_1'$,$\ell_2'$).}} \mid \ell_1 \preceq \ell_2\}$.
In the following, we argue that,
for any $\diamond^w_J\in M$ and \texttt{O}${}={}$\texttt{card}${}/{}$\texttt{incl}${}/{}$\texttt{pref}
if $\diamond={\leq}/{\subseteq}/{\preceq}$,
we have that $X \diamond^w_J Y$ holds iff \texttt{equal($J$,$w$,O)}${}\in Z$,
while $Y \diamond^w_J X$ does not hold for some $\diamond^w_J\in M$
iff \texttt{worse($J$)}${}\in Z$.
We thus consider the possible criteria for every $\diamond^w_J\in M$:
\begin{enumerate}
\item 
If $\diamond={\leq}$,
we have that atoms \texttt{count(S,$w$,Q,0$\dots$I)}
are derived via the rules in Line~7, 31--32, and 34--37 of \texttt{metaO.lp}, 
where \texttt{Q} is
the maximum index such that facts of the form \texttt{minimize($J$,S)} and
\texttt{wlist(S,Q,$\ell'$,W)} belong to~$F$.
Then, we have that~\texttt{I} is the cardinality of the multiset of
occurrences of $\ell=w@J$ in
$\minimize[\ell_1=w_1@J_1,\dots,\ell_k=w_k@J_k]$
such that $X\models\ell$.
Starting from~\texttt{I} in Line~39,
\texttt{cdown(S,$w$,-1,0)} is derived via the rules in Line~40--42 
iff
the cardinality of the multiset of
occurrences of $\ell=w@J$ in
$\minimize[\ell_1=w_1@J_1,\dots,\ell_k=w_k@J_k]$
such that $Y\models\ell$ is at least~\texttt{I},
and \texttt{cdown(S,$w$,-1,-1)} is derived via the same rules
iff
the cardinality of this multiset is greater than~\texttt{I}.
That is, $X \leq^w_J Y$ holds iff \texttt{equal($J$,$w$,card)}
is derived via the rule in Line~44,
and \texttt{worse($J$)} is derived via the rule in Line~46 
if $Y \nleq^w_J X$.
\item
If $\diamond={\subseteq}$,
an atom \texttt{ndiff($\ell_i'$)} such that
$\ell_i=w@J$ for some $1\leq i\leq k$ is derived via the
rules in Line~7 and 50--57 of \texttt{metaO.lp} iff
$X\not\models\ell_i$ or $Y\models\ell_i$.
Thus,
\texttt{equal($J$,$w$,incl)} is in turn derived via the rule in Line~59
iff $X\subseteq^w_J Y$ holds.
Moreover,
   \texttt{worse($J$)} is derived via the rules in Line~61--62 and~63--64 if
there is some $\ell_i=w@J$ for $1\leq i\leq k$
such that $Y\models\ell_i$ and $X\not\models\ell_i$,
i.e., if $Y\nsubseteq^w_J X$.
\item
If $\diamond={\preceq}$,
an atom \texttt{cando($\ell_i'$)}, \texttt{nocan($\ell_i'$)}, \texttt{condo($\ell_i'$)},
or \texttt{nocon($\ell_i'$)} such that
$\ell_i=w@J$ for some $1\leq i\leq k$ is derived via the
rules in Line~7 and 68--92 of \texttt{metaO.lp} iff
$X\models\ell_i$ and $Y\not\models\ell_i$,
$X\not\models\ell_i$ or $Y\models\ell_i$,
$X\not\models\ell_i$ and $Y\models\ell_i$, or, respectively,
$X\models\ell_i$ or $Y\not\models\ell_i$.
Via the rule in Line~94--95,
we further derive \texttt{cando(S,$w$,$\ell'$)} iff \texttt{minimize($J$,S)},
\texttt{wlist(S,\_,$\ell'$,$w$)}, and \texttt{wlist(S,\_,$\ell_1'$,$w$)} belong to~$F$
such that $\ell\preceq\ell_1$, $X\models\ell$, $Y\not\models\ell$, 
$X\not\models\ell_1$, and $Y\models\ell_1$.
Then, $X\preceq^w_J Y$ holds
if there is no literal $\ell_i$ such that 
$\ell_i=w@J$ for some $1\leq i\leq k$,
$\ell_i\preceq\ell$, $\ell\npreceq\ell_i$,
$Y\models\ell_i$, and $X\not\models\ell_i$, i.e.,
if \texttt{nocon($\ell_i'$)} holds for every~$\ell_i'$
such that \texttt{wlist(S,\_,$\ell_i'$,$w$)} and \texttt{prefer($\ell_i'$,$\ell'$)},
but not \texttt{prefer($\ell'$,$\ell_i'$)}, belong to~$F$.
Hence, \texttt{equal($J$,$w$,pref)} is derived via the rule in Line~103--104
iff $X\preceq^w_J Y$ holds.
Conversely,
a literal $\ell_i$ such that 
$\ell_i=w@J$ for some $1\leq i\leq k$ provides no indication for
$Y\preceq^w_J X$ if
$Y\not\models\ell_i$ or $X\models\ell_i$,
if there is no 
$\ell=w@J$ in
$\minimize[\ell_1=w_1@J_1,\linebreak[1]\dots,\linebreak[1]\ell_k=w_k@J_k]$
such that $\ell_i\preceq\ell$, $X\models\ell$, and $Y\not\models\ell$, or if
$\ell\preceq\ell_i$, $\ell_i\npreceq\ell$, $X\models\ell$, and $Y\not\models\ell$
for some $\ell=w@J$ in
$\minimize[\ell_1=\nolinebreak w_1@J_1,\linebreak[1]\dots,\linebreak[1]\ell_k=w_k@J_k]$.
Some of these cases applies iff
\texttt{nocon(S,$w$,$\ell_i'$)} is derived via the rules in Line~97--101,
so that 
\texttt{worse($J$)} is in turn derived via the rule in Line~106
if $Y\npreceq^w_J X$.
\end{enumerate}
As we have above investigated all rules that can possibly derive \texttt{worse($J$)}, 
i.e., the rules in Line~46, 61--64, and~106 of \texttt{metaO.lp},
and since some of them applies iff
$\diamond^w_J\in M$ such that
$Y \diamond^w_J X$ does not hold,
we now conclude that \texttt{worse($J$)}${}\in Z$ iff 
$Y \diamond^w_J X$ does not hold for some $\diamond^w_J\in M$.
In addition, the above cases yield that \texttt{equal($J$,$w$,O)}${}\in Z$
for \texttt{O}${}={}$\texttt{card}${}/{}$\texttt{incl}${}/{}$\texttt{pref} iff
$X \diamond^w_J Y$ holds for $\diamond={\leq}/{\subseteq}/{\preceq}$ such that $\diamond^w_J\in M$.
In view of the rules in Line~13--23,
we further conclude that
\texttt{bot}${}\in Z$ iff
there are no priority level~$J$ and weight~$w$ such that
$Y \diamond^{w'}_{J'} X$ holds for all $\diamond^{w'}_{J'}\in M$ where $J' \geq J$
and
$X \diamond^w_J Y$ does not hold for $\diamond^w_J\in M$.
That is, \texttt{bot}${}\in Z$ iff $X$ is an answer set of~$\Pi$ that is not dominated by
the answer set~$Y$ of~$\Pi$ wrt~$M$.
However, for any answer set~$Z$ of $\Pi_{\text{meta}}$,
saturation in Line~61--62 of \texttt{metaD.lp} admits
\text{\texttt{bot}}${}\in Z$ 
only if \texttt{bot} is derived wrt every answer set~$Y$ of~$\Pi$, i.e.,
if there is no answer set~$Y$ of~$\Pi$ that dominates~$X$ wrt~$M$.
Hence,
the integrity constraint in Line~27 of~\texttt{metaO.lp} selects
exactly the answer sets~$Z$ of $\Pi_{\text{meta}}\setminus\{\texttt{:- not bot.}\}$
such that $X=\{a \mid \text{\texttt{hold(atom($a$))}}\in Z\}$
is an undominated and thus optimal answer set of~$\Pi$ wrt~$M$.
Finally,
any such answer set~$Z$ of $\Pi_{\text{meta}}$ is unique since all atoms in~$Z$ 
that are not of the form \texttt{hold(atom($a$))}
follow deterministically once the interpretation of these atoms is fixed.
In particular,
\texttt{bot} and \texttt{true(atom($a$))} as well as \texttt{fail(atom($a$))} for all $a\in\atom{\Pi}$
hold in view of the integrity constraint in Line~27 of~\texttt{metaO.lp}
along with saturation in Line~61--62 of \texttt{metaD.lp}.
\end{proof}%
}%
{}

Regarding the computational complexity of tasks that can be addressed
using our meta-programming approach to optimization,
we first note that deciding whether there is an optimal answer set is in $\mathit{NP}$,
as the existence of some answer set (decidable by means of \texttt{meta.lp} only)
is sufficient for concluding that there is also an optimal one.
However, the inherent complexity becomes more sensible if we consider
the question of whether some atom~$a$ 
belongs to an optimal answer set.
To decide it, one can augment the reified input program
(but not the input program itself), \texttt{meta.lp}, \texttt{metaD.lp}, and \texttt{metaO.lp}
with the integrity constraint \texttt{:- not hold(atom($a$))}.
Then, several complex optimization criteria at a single priority level~\texttt{1} 
lead to completeness for~$\Sigma_2^P$, the second level of the polynomial time hierarchy,
thus showing that disjunctive ASP is appropriate to implement them.
To see this, note that deciding whether an atom~$a$ belongs to some
answer set of a positive disjunctive logic program is $\Sigma_2^P$-complete \cite{eitgot95a}.
When disjunctions $a_1\vee\dots\vee a_k$ in the heads of rules
are rewritten to $1\ \summ[a_1=1,\dots,a_k=1]$,
the question of whether an atom~$a$ belongs to an
answer set of the original program can be addressed by
reifying the rewritten program,
adding the integrity constraint \texttt{:- not hold(atom($a$))},
and applying \texttt{meta.lp}, \texttt{metaD.lp}, and \texttt{metaO.lp} wrt several optimization criteria.
For one, we can include a $\minimize$ statement over all atoms of the input program,
each associated with a different weight, to exploit the Pareto improvement implemented
in \texttt{metaO.lp} for refuting a candidate answer set including~$a$
if it does not correspond to a minimal model, i.e., an answer set
of the original program.
Alternatively, we can include a $\minimize$ statement over all atoms of the input program,
each having the weight~\texttt{1}, and augment the meta-program with the fact
\texttt{optimize(1,1,incl)}.
We could also use a $\minimize$ statement over all atoms~$a_i$ of the input program
along with their negation, each having the weight~\texttt{1}, 
and add the facts
 \texttt{optimize(1,1,pref)} as well as
\texttt{prefer(neg(atom($a_i$)),pos(atom($a_i$)))}. 
In view of these reductions, we conclude that
Pareto efficiency, inclusion, and literal preferences 
independently capture computational tasks located at the second level
of the polynomial time hierarchy, and our meta-programs allow for addressing
them via an extended program along with facts (and possibly also integrity constraints)
steering optimization relative to its reification.



%% file: application.tex
\section{Applications:  A Case Study}\label{sec:application}

While the approach of Eiter and Polleres \citeyear{eitpol06a}
consists of combining
two separate logic programs, one for ``guessing'' and a second
one for ``checking,'' into a disjunctive program addressing both tasks,
our meta-programming technique applies to a single (reified) input
program along with complex optimization criteria. 
In fact, we provide a generic implementation of such criteria
on top of extended programs encoding solution spaces.
Hence, our meta-programming technique allows for a convenient representation
of reasoning tasks 
in which testing the optimality of solutions to an
underlying problem in $\mathit{NP}$ lifts the complexity to $\Sigma_2^P$-hardness.
Respective formalisms include 
ordinary, parallel, as well as prioritized circumscription \cite{mccarthy80,lifschitz85b},
minimal consistency-based diagnosis \cite{reiter87b}, and
preferred extensions of argumentation frameworks \cite{dung95a}.
Similarly, Pareto efficiency is an important optimality condition in
decision making \cite{chenlama07a} and system design \cite{gries04a}.
In the following, we illustrate the application of our approach
on the example of an existing real-world application:
repair wrt large gene-regulatory networks \cite{geguivscsithve10a}.

\begin{lstlisting}[float=t,caption={Encoding of repair wrt regulatory networks and experiment data (\texttt{\small repair.lp}).},captionpos=b,label=fig:repair]
% auxiliary concepts

sign(-1;1).
complement(S,-S) :- sign(S).

% construct candidate repair

pos(aedge(U,V),  J,W) :- repair(aedge,J,W), vertex(U;V), U != V.
pos(eflip(U,V,S),J,W) :- repair(eflip,J,W), obs_elabel(U,V,S).
pos(ivert(V),    J,W) :- repair(ivert,J,W), vertex(V).
pos(pvert(P,V),  J,W) :- repair(pvert,J,W), exp(P), vertex(V).
pos(vflip(P,V,S),J,W) :- repair(vflip,J,W), obs_vlabel(P,V,S).

{ apply(R) } :- pos(R,_,_).

% construct consistent total labelings

elabel(U,V,S) :- not apply(eflip(U,V,S)), obs_elabel(U,V,S).
elabel(U,V,T) :-     apply(eflip(U,V,S)),                    complement(S,T).
elabel(U,V,S) :-     apply(aedge(U,V)),   not elabel(U,V,T), complement(S,T).
elabel(U,V,S) :-     edge(U,V),           not elabel(U,V,T), complement(S,T).

vlabel(P,V,S) :- not apply(vflip(P,V,S)), obs_vlabel(P,V,S).
vlabel(P,V,T) :-     apply(vflip(P,V,S)), complement(S,T).
vlabel(P,V,S) :- not vlabel(P,V,T),       complement(S,T),   exp(P), vertex(V).

inf(P,V,S*T)  :- elabel(U,V,S), vlabel(P,U,T), not inp(P,V).

:- vlabel(P,V,S), not inf(P,V,S), not inp(P,V), not apply(ivert(V);pvert(P,V)).

% optimize repair

#minimize [ apply(R) = W @ J : pos(R,J,W) ].
\end{lstlisting}
Listing~\ref{fig:repair} shows a simplified version of the repair encoding
given in \cite{geguivscsithve10a}.
It applies to a regulatory network, a directed graph with (partially) labeled edges,
represented by facts of the predicates \texttt{vertex/1}, \texttt{edge/2},
and \texttt{obs\_elabel/3}, where a label~\texttt{S} is 
\texttt{1} (activation) or \texttt{-1} (inhibition).
In addition, the data of experiments labeled~\texttt{P}
are provided by facts of the predicates \texttt{exp/1}, \texttt{inp/2}
denoting input vertices (subject to perturbations), and
\texttt{obs\_vlabel/3}, where a label~\texttt{S} is again
\texttt{1} (increase) or \texttt{-1} (decrease).
The regulatory network is consistent with the experiment data if
there are total labelings of edges and vertices (for each experiment~\texttt{P})
such that the label of every non-input vertex~\texttt{V} is explained by the 
influence of some of its regulators~\texttt{U},
where the influence is the product \texttt{S*T} of the edge label~\texttt{S}
and the label~\texttt{T} of~\texttt{U} (in experiment~\texttt{P}).
In the practice of systems biology, regulatory networks and 
experiment data \EXT{are often collected from heterogeneous sources,
and likewise incomplete or noisy.
In view of this,
it is very likely that a regulatory network and experiment data are}{often turn out to be} mutually inconsistent,
which makes it highly non-trivial to draw biologically 
meaningful conclusions in an automated way.
To address this shortage,
several repair operations were devised in \cite{geguivscsithve10a},
which can be enabled via facts of the form \texttt{repair(K,J,W)},
where \texttt{K} indicates a certain kind of admissible repair operations,
\texttt{J} a priority level, and
\texttt{W} a weight.
The repair operations~\texttt{R} to apply
are selected via the rule in Line~14 of Listing~\ref{fig:repair}, and
their effects are propagated via the rules in Line~18--29,
thus obtaining total edge and vertex labelings
witnessing the reestablishment of consistency.
Given that applications of repair operations modify a
regulatory network or experiment data\EXT{ (depending on the kind of operations)}{},
we are interested in applying few operations only,
which is expressed by the $\minimize$ statement in Line~33.

A reasonable repair configuration could consist of facts of the following form:
\begin{description}
\item[\textnormal{\texttt{repair(ivert,J$_1$,W$_1$).}}]
admitting to turn vertices into inputs in \emph{all} experiments.
\item[\textnormal{\texttt{repair(eflip,J$_2$,W$_2$).}}]
admitting network modifications by flipping edge labels.
\item[\textnormal{\texttt{repair(pvert,J$_3$,W$_3$).}}]
admitting to turn vertices into inputs in specific experiments.
\item[\textnormal{\texttt{repair(vflip,J$_4$,W$_4$).}}]
admitting data modifications by flipping vertex labels.
\end{description}
While the kinds of repair referred to by \texttt{ivert} and \texttt{eflip}
operate primarily on a network (in view of incompleteness or incorrectness),
the ones denoted by \texttt{pvert} and \texttt{vflip} mainly address the data
(which can be noisy).
If we penalize all repair operations uniformly via
\texttt{J}${}={}$\texttt{J}$_1={}$\texttt{J}$_2={}$\texttt{J}$_3={}$\texttt{J}$_4$ and
\texttt{W}${}={}$\texttt{W}$_1={}$\texttt{W}$_2={}$\texttt{W}$_3={}$\texttt{W}$_4$,
the instantiation of the $\minimize$ statement in Line~33
represents ordinary cardinality-based optimization,
assembled in solvers like \clasp\ 
and \smodels. 
However, by adding \texttt{optimize(J,W,incl)} as a fact,
we can easily switch to inclusion-based minimization and
use a disjunctive solver like \claspD\ 
to solve the more complex problem.
While our meta-programs enable such a shift of optimization criteria by means of 
adding just one fact,
a direct disjunctive encoding of inclusion-based minimization has been provided in 
\cite{geguivscsithve10a}; note that the latter is by far more involved than the
basic  repair encoding in Listing~\ref{fig:repair}.
Furthermore,
our meta-programming approach allows us to distinguish between
different kinds of repair operations (without prioritizing them)
and optimize wrt Pareto efficiency.
To accomplish this,
one only needs to pick unequal values for 
\texttt{W}$_1,\dots,{}$\texttt{W}$_4$,
where cardinality-based minimization wrt each \texttt{W}$_i$ can selectively be 
replaced by inclusion via providing a fact \texttt{optimize(J,W$_i$,incl)}.
Finally, we can choose to rank kinds of repair operations
by providing different priority levels \texttt{J}$_1,\dots,{}$\texttt{J}$_4$.
\EXT{%
For instance, network repairs may be considered as ``more drastic'' operations
than data repairs, which can be expressed by assigning levels such that
\texttt{J}$_1={}$\texttt{J}$_2>{}$\texttt{J}$_3={}$\texttt{J}$_4$.
In fact, the kinds of admissible repair operations as well as appropriate penalties
for applying them depend on the biological application and experience. }{}%
In this respect, the flexibility gained due to meta-programming allows for deploying and
comparing different optimization criteria, e.g., regarding the
accuracy of resulting predictions (cf.\ \cite{geguivscsithve10a}).

For giving an account of the practical capabilities of our meta-programming approach,
we empirically compared it to the direct encoding of inclusion-based minimization in
\cite{geguivscsithve10a}.
To this end, we ran \gringo\ version 3.0.3 and \claspD\ version 1.1
on 100 instances
\EXT{(random samples from real data of two biological experiments,
as also used in \cite{geguivscsithve10a}) }{}%
wrt three kinds of admissible repair operations,
resulting in 300 runs each with our meta-programs and with the direct encoding.
All runs have been performed sequentially on a machine equipped with Intel Xeon E5520 processors
and 48~GB main memory under Linux, imposing a time limit of 4000 sec per run.
\EXT{(In view of moderate memory consumption, a space limit was not needed.) }{}%
To our own surprise,
more runs were completed in time with the meta-programs than with the direct encoding:
219 versus 150.%
\footnote{All instances and detailed results are available at \cite{metasp}.}
The disadvantages of the direct encoding show that further gearing would be
required to improve solving efficiency,
which adds to the difficulty of furnishing a functional saturation-based encoding.
In view of this, we conclude that our meta-programming approach to complex optimization
is an eligible and viable alternative.
However,
enhancements of disjunctive ASP solvers boosting its performance
would still be desirable.



%% file: discussion.tex
\section{Discussion}\label{sec:discussion}

\EXT{%
The major contribution of this work is the provision of easy modeling capacities
for expressing complex preferences involving a higher computational complexity.
We accomplish this by developing advanced meta-modeling techniques that allow us
to reinterpret $\minimize$ statements by means of ASP.
This methodology offers a generic saturation-based implementation technique for
testing     various optimality conditions wrt answer sets of extended logic programs,
as used with \lparse\ and \gringo.
At \cite{metasp},
we provide enhanced versions of our meta-encodings,
supporting additional features of \gringo's input language \cite{potasscoManual}
as well as some basic optimizations with regard to grounding performance.
These encodings provide a ready-to-use platform for implementing
applications involving complex optimization criteria in ASP.}{}

Our integral 
approach to modeling complex optimization criteria in ASP
brings about a number of individual contributions.
To begin with, we introduce the reification capacities of our grounder \gringo\
along with the associated meta-encoding,
paving the way to the immediate availability of meta-modeling techniques.
In fact, the full version of the basic meta-encoding in
Listing~\ref{fig:reified}\EXT{}{, obtainable at \cite{metasp},} covers the complete
language of \gringo, including disjunctions and diverse aggregates. 
Moreover,
our meta-modeling techniques
provide a general account of saturation and, 
thus, abolish its compulsory replication
for expressing complex preferences.
Of particular interest is the stability property of answer sets serving as
implicit counterexamples.
Unlike the approach of Eiter and Polleres \citeyear{eitpol06a},
our encoding avoids ``guessing'' level mappings.
Also, our target language involves choice rules and $\summ$ constraints
\cite{siniso02a}, and we are unaware of any pre-existing meta-encoding  of
corresponding answer sets,
neither as candidates nor as counterexamples. 
Likewise, related meta-programming approaches for generating 
consequences of logic programs \cite{fabwol09a}
or
explanations wrt debugging queries \cite{oeputo10a}
do not consider such aggregates (but disjunctions in object programs).

We exploit the two-dimensionality of $\minimize$ statements by using levels and
weights for combining a lexicographic ranking with Pareto efficiency.
At each level, \EXT{one may choose whether }{}groups of literals sharing the same
weight \EXT{are}{can be} compared wrt \EXT{cardinality or }{}inclusion.
\EXT{Also, summation can be modeled by exploiting the multiset property of $\minimize$ statements. }{}%
This is extended \EXT{by }{in \cite{gekasc11A} by cardinality-based minimization and }the framework of \cite{sakino00a}, relying
on a preference relation among literals (given in addition to $\minimize$ statements)\EXT{}{;
the augmented encoding is also available at \cite{metasp}}.
In fact, the approach of Section~\ref{sec:approach} allows for capturing the
special cases of parallel and prioritized circumscription, investigated by Janhunen
and Oikarinen in~\citeyear{janoik04a,oikjan08b}.
An interesting future extension is the encoding of 
optimality conditions for logic programs with ordered
disjunction \cite{brnisy04a}, whose custom-made implementation in the prototype \psmodels\
interleaves two \smodels\ oracles for accomplishing a generate-and-test approach
similar to the idea of our meta-programs.
Ultimately, our approach could serve as an implementation platform for answer
set optimization \cite{brnitr03a} and the preference description language
proposed in \cite{brewka04b}.
Last but not least, our meta-programs furnish a rich and readily available source
of $\Sigma_2^P$-hard challenge problems, fostering the development of ASP solvers for
disjunctive logic programming.


%% file: acknowledgments.tex
\EXT{\paragraph{Acknowledgments.}}
{\smallskip\noindent\textit{Acknowledgments. }}
This work was partly funded by DFG grant SCHA 550/8-2.
We are grateful to Tomi Janhunen, Ilkka Niemel{\"a}, and the referees
for their helpful comments.


%% file: bbl.tex

%% file: paper.bbl
\begin{thebibliography}{}

\bibitem[\protect\citeauthoryear{Baral}{Baral}{2003}]{baral02a}
{\sc Baral, C.} 2003.
\newblock {\em Knowledge Representation, Reasoning and Declarative Problem
  Solving}.
\newblock Cambridge University Press.


\bibitem[\protect\citeauthoryear{Brewka}{Brewka}{2004}]{brewka04b}
{\sc Brewka, G.} 2004.
\newblock Answer sets: From constraint programming towards qualitative
  optimization.
\newblock In {\em Proceedings of the Seventh International Conference on Logic
  Programming and Nonmonotonic Reasoning (LPNMR'04)}, {V.~Lifschitz} {and}
  {I.~Niemel{\"a}}, Eds. Lecture Notes in Artificial Intelligence, vol. 2923.
  Springer-Verlag, 34--46.

\bibitem[\protect\citeauthoryear{Brewka, Niemel{\"a}, and Syrj{\"a}nen}{Brewka
  et~al\mbox{.}}{2004}]{brnisy04a}
{\sc Brewka, G.}, {\sc Niemel{\"a}, I.}, {\sc and} {\sc Syrj{\"a}nen, T.} 2004.
\newblock Logic programs with ordered disjunction.
\newblock {\em Computational Intelligence\/}~{\em 20,\/}~2, 335--357.

\bibitem[\protect\citeauthoryear{Brewka, Niemel{\"a}, and Truszczynski}{Brewka
  et~al\mbox{.}}{2003}]{brnitr03a}
{\sc Brewka, G.}, {\sc Niemel{\"a}, I.}, {\sc and} {\sc Truszczynski, M.} 2003.
\newblock Answer set optimization.
\newblock In {\em Proceedings of the Eighteenth International Joint Conference
  on Artificial Intelligence (IJCAI'03)}, {G.~Gottlob} {and} {T.~Walsh}, Eds.
  Morgan Kaufmann Publishers, 867--872.

\bibitem[\protect\citeauthoryear{Chevaleyre, Endriss, Lang, and
  Maudet}{Chevaleyre et~al\mbox{.}}{2007}]{chenlama07a}
{\sc Chevaleyre, Y.}, {\sc Endriss, U.}, {\sc Lang, J.}, {\sc and} {\sc Maudet,
  N.} 2007.
\newblock A short introduction to computational social choice.
\newblock In {\em Proceedings of the Thirty-third Conference on Current Trends in
  Theory and Practice of Computer Science (SOFSEM'07)}, {J.~{van Leeuwen}},
  {G.~Italiano}, {W.~{van der Hoek}}, {C.~Meinel}, {H.~Sack}, {and}
  {F.~Plasil}, Eds. Lecture Notes in Computer Science, vol. 4362.
  Springer-Verlag, 51--69.

\bibitem[\protect\citeauthoryear{Delgrande, Schaub, and Tompits}{Delgrande
  et~al\mbox{.}}{2003}]{descto02a}
{\sc Delgrande, J.}, {\sc Schaub, T.}, {\sc and} {\sc Tompits, H.} 2003.
\newblock A framework for compiling preferences in logic programs.
\newblock {\em Theory and Practice of Logic Programming\/}~{\em 3,\/}~2, 
  129--187.

\bibitem[\protect\citeauthoryear{Drescher, Gebser, Grote, Kaufmann, K{\"o}nig,
  Ostrowski, and Schaub}{Drescher et~al\mbox{.}}{2008}]{drgegrkakoossc08a}
{\sc Drescher, C.}, {\sc Gebser, M.}, {\sc Grote, T.}, {\sc Kaufmann, B.}, {\sc
  K{\"o}nig, A.}, {\sc Ostrowski, M.}, {\sc and} {\sc Schaub, T.} 2008.
\newblock Conflict-driven disjunctive answer set solving.
\newblock In {\em Proceedings of the Eleventh International Conference on
  Principles of Knowledge Representation and Reasoning (KR'08)}, {G.~Brewka}
  {and} {J.~Lang}, Eds. AAAI Press, 422--432.

\bibitem[\protect\citeauthoryear{Dung}{Dung}{1995}]{dung95a}
{\sc Dung, P.} 1995.
\newblock On the acceptability of arguments and its fundamental role in
  nonmonotonic reasoning, logic programming and $n$-person games.
\newblock {\em Artificial Intelligence\/}~{\em 77,\/}~2, 321--357.

\bibitem[\protect\citeauthoryear{Eiter, Faber, Leone, and Pfeifer}{Eiter
  et~al\mbox{.}}{2003}]{eifalepf03a}
{\sc Eiter, T.}, {\sc Faber, W.}, {\sc Leone, N.}, {\sc and} {\sc Pfeifer, G.}
  2003.
\newblock Computing preferred answer sets by meta-interpretation in answer set
  programming.
\newblock {\em Theory and Practice of Logic Programming\/}~{\em 3,\/}~4-5,
  463--498.

\bibitem[\protect\citeauthoryear{Eiter and Gottlob}{Eiter and
  Gottlob}{1995}]{eitgot95a}
{\sc Eiter, T.} {\sc and} {\sc Gottlob, G.} 1995.
\newblock On the computational cost of disjunctive logic programming:
  Propositional case.
\newblock {\em Annals of Mathematics and Artificial Intelligence\/}~{\em
  15,\/}~3-4, 289--323.

\bibitem[\protect\citeauthoryear{Eiter and Polleres}{Eiter and
  Polleres}{2006}]{eitpol06a}
{\sc Eiter, T.} {\sc and} {\sc Polleres, A.} 2006.
\newblock Towards automated integration of guess and check programs in answer
  set programming: a meta-interpreter and applications.
\newblock {\em Theory and Practice of Logic Programming\/}~{\em 6,\/}~1-2,
  23--60.

\bibitem[\protect\citeauthoryear{Faber and Woltran}{Faber and
  Woltran}{2009}]{fabwol09a}
{\sc Faber, W.} {\sc and} {\sc Woltran, S.} 2009.
\newblock Manifold answer-set programs for meta-reasoning.
\newblock In {\em Proceedings of the Tenth International Conference on Logic
  Programming and Nonmonotonic Reasoning (LPNMR'09)}, {E.~Erdem}, {F.~Lin},
  {and} {T.~Schaub}, Eds. Lecture Notes in Artificial Intelligence, vol. 5753.
  Springer-Verlag, 115--128.

\bibitem[\protect\citeauthoryear{Fages}{Fages}{1994}]{fages94a}
{\sc Fages, F.} 1994.
\newblock Consistency of {C}lark's completion and the existence of stable
  models.
\newblock {\em Journal of Methods of Logic in Computer Science\/}~{\em 1},
  51--60.

\bibitem[\protect\citeauthoryear{Ferraris}{Ferraris}{2005}]{ferraris05a}
{\sc Ferraris, P.} 2005.
\newblock Answer sets for propositional theories.
\newblock In {\em Proceedings of the Eighth International Conference on Logic
  Programming and Nonmonotonic Reasoning (LPNMR'05)}, {C.~Baral}, {G.~Greco},
  {N.~Leone}, {and} {G.~Terracina}, Eds. Lecture Notes in Artificial
  Intelligence, vol. 3662. Springer-Verlag, 119--131.

\bibitem[\protect\citeauthoryear{Garey and Johnson}{Garey and
  Johnson}{1979}]{garjoh79}
{\sc Garey, M.} {\sc and} {\sc Johnson, D.} 1979.
\newblock {\em Computers and Intractability: A Guide to the Theory of
  {NP}-Completeness}.
\newblock W. Freeman and Co. 

\bibitem[\protect\citeauthoryear{Gebser, Guziolowski, Ivanchev, Schaub, Siegel,
  Thiele, and Veber}{Gebser et~al\mbox{.}}{2010}]{geguivscsithve10a}
{\sc Gebser, M.}, {\sc Guziolowski, C.}, {\sc Ivanchev, M.}, {\sc Schaub, T.},
  {\sc Siegel, A.}, {\sc Thiele, S.}, {\sc and} {\sc Veber, P.} 2010.
\newblock Repair and prediction (under inconsistency) in large biological
  networks with answer set programming.
\newblock In {\em Proceedings of the Twelfth International Conference on
  Principles of Knowledge Representation and Reasoning (KR'10)}, {F.~Lin} {and}
  {U.~Sattler}, Eds. AAAI Press, 497--507.

\bibitem[\protect\citeauthoryear{Gebser, Kaminski, Kaufmann, Ostrowski, Schaub,
  and Thiele}{Gebser et~al\mbox{.}}{\!\!}]{potasscoManual}
{\sc Gebser, M.}, {\sc Kaminski, R.}, {\sc Kaufmann, B.}, {\sc Ostrowski, M.},
  {\sc Schaub, T.}, {\sc and} {\sc Thiele, S.}
\newblock A user's guide to \texttt{gringo}, \texttt{clasp}, \texttt{clingo},
  and \texttt{iclingo}.\par
\newblock Available at {\texttt{http://potassco.sourceforge.net}}.

\EXT{}{%
\bibitem[\protect\citeauthoryear{Gebser, Kaminski, and Schaub}{Gebser
  et~al\mbox{.}}{2011}]{gekasc11A}
{\sc Gebser, M.}, {\sc Kaminski, R.}, {\sc and} {\sc Schaub, T.} 2011.
\newblock Complex optimization in answer set programming: Extended version.
\newblock Available at \cite{metasp}.\par
\newblock ({\em This is an extended version of the paper at hand.})}

\bibitem[\protect\citeauthoryear{Gebser, Kaufmann, Neumann, and Schaub}{Gebser
  et~al\mbox{.}}{2007}]{gekanesc07a}
{\sc Gebser, M.}, {\sc Kaufmann, B.}, {\sc Neumann, A.}, {\sc and} {\sc Schaub,
  T.} 2007.
\newblock Conflict-driven answer set solving.
\newblock In {\em Proceedings of the Twentieth International Joint Conference
  on Artificial Intelligence (IJCAI'07)}, {M.~Veloso}, Ed. AAAI Press/The MIT
  Press, 386--392.

\bibitem[\protect\citeauthoryear{Gelfond and Lifschitz}{Gelfond and
  Lifschitz}{1991}]{gellif91a}
{\sc Gelfond, M.} {\sc and} {\sc Lifschitz, V.} 1991.
\newblock Classical negation in logic programs and disjunctive databases.
\newblock {\em New Generation Computing\/}~{\em 9}, 365--385.

\bibitem[\protect\citeauthoryear{Gries}{Gries}{2004}]{gries04a}
{\sc Gries, M.} 2004.
\newblock Methods for evaluating and covering the design space during early
  design development.
\newblock {\em Integration\/}~{\em 38,\/}~2, 131--183.

\bibitem[\protect\citeauthoryear{Janhunen and Oikarinen}{Janhunen and
  Oikarinen}{2004}]{janoik04a}
{\sc Janhunen, T.} {\sc and} {\sc Oikarinen, E.} 2004.
\newblock Capturing parallel circumscription with disjunctive logic programs.
\newblock In {\em Proceedings of the Ninth European Conference on Logics in
  Artificial Intelligence (JELIA'04)}, {J.~Alferes} {and} {J.~Leite}, Eds.
  Lecture Notes in Computer Science, vol. 3229. Springer-Verlag, 134--146.

\bibitem[\protect\citeauthoryear{Leone, Pfeifer, Faber, Eiter, Gottlob, Perri,
  and Scarcello}{Leone et~al\mbox{.}}{2006}]{dlv03a}
{\sc Leone, N.}, {\sc Pfeifer, G.}, {\sc Faber, W.}, {\sc Eiter, T.}, {\sc
  Gottlob, G.}, {\sc Perri, S.}, {\sc and} {\sc Scarcello, F.} 2006.
\newblock The {DLV} system for knowledge representation and reasoning.
\newblock {\em ACM Transactions on Computational Logic\/}~{\em 7,\/}~3,
  499--562.

\bibitem[\protect\citeauthoryear{Lifschitz}{Lifschitz}{1985}]{lifschitz85b}
{\sc Lifschitz, V.} 1985.
\newblock Computing circumscription.
\newblock In {\em Proceedings of the Ninth International Joint Conference on
  Artificial Intelligence  (IJCAI'85)}, {A.~Joshi}, Ed. Morgan Kaufmann Publishers,
  121--127.

\bibitem[\protect\citeauthoryear{Liu and You}{Liu and You}{2010}]{liuyou10a}
{\sc Liu, G.} {\sc and} {\sc You, J.} 2010.
\newblock Level mapping induced loop formulas for weight constraint and
  aggregate logic programs.
\newblock {\em Fundamenta Informaticae\/}~{\em 101,\/}~3, 237--255.

\bibitem[\protect\citeauthoryear{Lloyd}{Lloyd}{1987}]{lloyd87}
{\sc Lloyd, J.} 1987.
\newblock {\em Foundations of Logic Programming\/}, 2nd ed.
\newblock Symbolic Computation. Springer-Verlag.

\bibitem[\protect\citeauthoryear{McCarthy}{McCarthy}{1980}]{mccarthy80}
{\sc McCarthy, J.} 1980.
\newblock Circumscription --- a form of nonmonotonic reasoning.
\newblock {\em Artificial Intelligence\/}~{\em 13,\/}~1-2, 27--39.

\bibitem[\protect\citeauthoryear{metasp}{metasp}{\!\!}]{metasp}
\textsc{metasp}.
\newblock \texttt{http://www.cs.uni-potsdam.de/wv/metasp}.

\bibitem[\protect\citeauthoryear{Oetsch, P{\"u}hrer, and Tompits}{Oetsch
  et~al\mbox{.}}{2010}]{oeputo10a}
{\sc Oetsch, J.}, {\sc P{\"u}hrer, J.}, {\sc and} {\sc Tompits, H.} 2010.
\newblock Catching the ouroboros: On debugging non-ground answer-set programs.
\newblock 
  {\em Theory and Practice of Logic Programming. Twenty-sixth
  International Conference on Logic Programming (ICLP'10) Special Issue\/}~{\em 10,\/}~4-6, 513--529.

\bibitem[\protect\citeauthoryear{Oikarinen and Janhunen}{Oikarinen and
  Janhunen}{2008}]{oikjan08b}
{\sc Oikarinen, E.} {\sc and} {\sc Janhunen, T.} 2008.
\newblock Implementing prioritized circumscription by computing disjunctive
  stable models.
\newblock In {\em Proceedings of the Thirteenth International Conference on
  Artificial Intelligence: Methodology, Systems, and Applications (AIMSA'08)},
  {D.~Dochev}, {M.~Pistore}, {and} {P.~Traverso}, Eds. Lecture Notes in
  Artificial Intelligence, vol. 5253. Springer-Verlag, 167--180.

\bibitem[\protect\citeauthoryear{Reiter}{Reiter}{1987}]{reiter87b}
{\sc Reiter, R.} 1987.
\newblock A theory of diagnosis from first principles.
\newblock {\em Artificial Intelligence\/}~{\em 32,\/}~1, 57--96.

\bibitem[\protect\citeauthoryear{Sakama and Inoue}{Sakama and
  Inoue}{2000}]{sakino00a}
{\sc Sakama, C.} {\sc and} {\sc Inoue, K.} 2000.
\newblock Prioritized logic programming and its application to commonsense
  reasoning.
\newblock {\em Artificial Intelligence\/}~{\em 123,\/}~1-2, 185--222.

\bibitem[\protect\citeauthoryear{Simons, Niemel{\"a}, and Soininen}{Simons
  et~al\mbox{.}}{2002}]{siniso02a}
{\sc Simons, P.}, {\sc Niemel{\"a}, I.}, {\sc and} {\sc Soininen, T.} 2002.
\newblock Extending and implementing the stable model semantics.
\newblock {\em Artificial Intelligence\/}~{\em 138,\/}~1-2, 181--234.

\bibitem[\protect\citeauthoryear{Syrj{\"a}nen}{Syrj{\"a}nen}{\!\!}]{lparseManual}
{\sc Syrj{\"a}nen, T.}
\newblock Lparse 1.0 user's manual.\par
\newblock Available at {\texttt{http://www.tcs.hut.fi/Software/smodels/lparse.ps.gz}}.

\end{thebibliography}
